\documentclass[10pt]{iopart}

    \usepackage{iopams}
    \usepackage{amssymb}
    \usepackage{amsfonts}
    \usepackage{amstext}
    \usepackage{graphicx}
    \usepackage{setstack}
    \usepackage{amscd}
    \usepackage{color}
    \usepackage{ulem}

    \newcommand{\ihbar}{\imath \hbar}
    \newcommand{\Pe}{\mathbb{P}e_{\leftarrow}}
    \newcommand{\aPe}{\mathbb{P}e_{\rightarrow}}
    \newcommand{\llangle}{\langle \hspace{-0.2em} \langle}
    \newcommand{\rrangle}{\rangle \hspace{-0.2em} \rangle}
    \renewcommand{\S}{{\mathcal H_{\mathcal S}}}
    \newcommand{\E}{{\mathcal H_{\mathcal E}}}
    \newcommand{\SE}{{\mathcal H_{\mathcal S} \otimes \mathcal H_{\mathcal E}}}
    \newcommand{\iint}{\int \!\!\!  \int}
    \newcommand{\card}{\mathrm{card} }
    \newcommand{\Ran}{\mathrm{Ran}}

    \newenvironment{proof}{\noindent \textit{Proof:}}{\hfill $\Box$ \\}

    \newtheorem{prop}{Property}

\begin{document}

    \title{Adiabatic quantum control hampered by entanglement}

    \author{David Viennot}

    \address{Institut UTINAM (CNRS UMR 6213, Universit\'e de Franche-Comt\'e), 41 bis Avenue de
    l'Observatoire, BP1615, 25010 Besan{\c c}on cedex, France}

		\eads{
		\mailto{david.viennot@utinam.cnrs.fr},
		}

\begin{abstract}
We study the defects in adiabatic control of a quantum system caused by the entanglement of the system with its environment. Such defects can be assimilated to decoherence processes due to perturbative couplings between the system and the environment. To analyse these effects, we propose a geometric approach, based on a field theory on the control manifold issued from the higher gauge theory associated with the $C^*$-geometric phases. We study a visualization method to analyse the defects of the adiabatic control based on the drawing of the field strengths of the gauge theory. To illustrate the present methodology we consider the example of the atomic STIRAP (stimulated Raman adiabatic passage) where the controlled atom is entangled with another atom. We study the robustness of the STIRAP effect when the controlled atom is entangled with another one.
\end{abstract}

\section{Introduction}
Quantum control is one of the main research subject of the modern physics. A quantum control problem consists to find how vary the external parameters (as for example parameters of strong laser fields) in order to a quantum system (a spin, an atom or a molecule) evolves to a predetermined target state satisfying the control goal. Such problems can present very important applications in different fields: nanosciences (to drive molecular machines), quantum information (to perform quantum logic gates), and physical chemistry (to perform vibrational cooling, to control chemical reactions). If the target state is an eigenstate of the quantum system, an adiabatic approach \cite{Messiah} is a good strategy to solve quantum control problems, since external parameters are usually slowly varied with respect to the response of the quantum system and since the adiabatic approximation predicts that the wave function remains projected onto an eigenstate during the dynamics. Adiabatic schemes of quantum control have been proposed for quantum computation by holonomic approaches \cite{Zanardi, Lucarelli} or by quantum annealing \cite{Santoro}, and for atomic control by strong laser fields \cite{Guerin}.\\
Real quantum systems are never isolated. The coupling between the quantum system and its environment, even if it is perturbative, can induce defects of the control result with respect to the idealized isolated quantum system. The system and the environment are entangled and the dynamics can be characterized by decoherence processes. The goal of the present paper is the characterization of these defects. Some adiabatic quantum control methods are based on geometric approaches (geometry of fiber bundles \cite{Lucarelli} or topology of eigensurfaces \cite{Guerin}). We want a geometric characterization of the defects induced by the entanglement. Recently, we have proposed a generalization of the geometric phase concept for open and composite quantum systems \cite{Viennot1, Viennot2}. This geometric phase takes its values in the $C^*$-algebra of the operators of the quantum system. In contrast with the usual geometric phase which is associated with a simple gauge theory \cite{Simon}, the $C^*$-geometric phase is associated with an higher gauge theory (a generalization of gauge theory into a category theory context, see for example \cite{Baez}). This higher gauge theory involves some fields on the manifold spanned by the control external parameters, which are introduced in \cite{Viennot1}. In this paper, we do not want to recall the mathematical structure associated with this higher gauge theory (it can be found in \cite{Viennot1}), but we would interpret physically the implicated fields from the viewpoint of the quantum control hampered by entanglement.\\
This paper is organized as follows. Section 2 is devoted to a small review of the adiabatic quantum control for an idealized isolated quantum system. The goal of this section is the introduction of some notations and concepts. Section 3 studies the theoretical properties of the fields associated with the higher gauge theory for the control hampered by entanglement. We study the physical meaning of these fields with respect to the quantum control problem. Section 4 is devoted to a simple but instructive example. The stimulated Raman adiabatic passage (STIRAP) is a solution of a quantum control problem consisting to change the state of a three level atom from the bound state to an excited state by passing through a ``dark'' state, by using two laser Gaussian pulses. We study the robustness of this solution when the atom is entangled with another one which feels the laser fields. We interpret the results by drawing the field strengths of the higher gauge theory to confirm the physical meanings of these fields. Section 5 is a discussion concerning the application of the methodology presented in this paper on a system entangled with a larger environment (by using a small effective Hamiltonians to represent this environment or by working only at the stage of the density matrices). In this paper we use the word ``environment'' with a large meaning. It can signify the large environment of an open quantum system as for the discussion section 5, or it can signify the (small) second part of a bipartite quantum system as for the example treated section 4. We focus on the effect of the entanglement on the control. We have chosen an example with a small environment (bipartite quantum system) in order to avoid the possible effects of the dissipation induced by large environments and enlighten the effects associated only with the entanglement.

\section{Adiabatic quantum control}
In this section, we consider an idealized isolated quantum system controlled by external parameters denoted by $x$. The set of all configurations of $x$ is assumed to form a $\mathcal C^\infty$-manifold $M$ called the control manifold. We denote by $\S$ the Hilbert space of the states of the system (for the sake of simplicity, in the whole of this paper we suppose that $\S$ is finite dimensional). The dynamics of the controlled quantum system is governed by the self-adjoint Hamiltonian $H(x) \in \mathcal L(\S)$ ($\mathcal L(\S)$ denotes the $C^*$-algebra of the operators of $\S$). A control solution is a path $\mathcal C: t \mapsto x(t) \in M$, such that the wave function $\psi(t) \in \S$ solution of the Schr\"odinger equation
\begin{equation}
\ihbar \frac{d\psi}{dt} = H(x(t)) \psi(t) \qquad \psi(0) = \psi_0 \in \S
\end{equation}
becomes, at $t=T$ the end of the control, $\psi(T) = \psi_{target}$. $\psi_{target} \in \S$ is the predetermined target state satisfying the goal of the control. We will suppose that the path $\mathcal C$ is closed (this is a generic situation, we start and we stop with a control system off).

\subsection{Adiabatic approximation}
Let $\{\lambda_a(x)\}_a$ be the instantaneous eigenvalues of $H(x)$ that we suppose as being non-degenerate for all $x \in M$ except eventually for some isolated points in $M$ in a first time. Let $\{\phi_a(x)\}_a$ be the associated normalized eigenvectors.
\begin{equation}
H(x) \phi_a(x) = \lambda_a(x) \phi_a(x)
\end{equation}
If $\psi(0) = \phi_a(x(0))$ and if 
\begin{equation}
\label{adiabcond}
\forall b\not= a, \quad \hbar \sup_{t \in [0,T]} \left| \frac{\langle \phi_a| \partial_\mu \phi_b \rangle \dot x^\mu}{\lambda_b - \lambda_a} \right| \ll 1
\end{equation}
then the wave function satisfies the adiabatic approximation (see for example \cite{Messiah})
\begin{equation}
\psi(T) \simeq e^{-\ihbar^{-1} \int_0^T \lambda_a(x(t))dt} e^{-\oint_{\mathcal C} A_a(x)} \phi_a(x(T))
\end{equation}
where the geometric phase $e^{-\oint_{\mathcal C} A_a(x)}$ discovered by Berry \cite{Berry} is generated by
\begin{equation}
A_a(x) = \langle \phi_a(x)|d\phi_a(x) \rangle \in \Omega^1 M
\end{equation}
$d$ is the exterior differential of $M$ and $\Omega^n M$ denotes the set of differential $n$-forms of $M$. The adiabatic condition (\ref{adiabcond}) implies that the control parameter variations ($\dot x^\mu$, the dot denotes the time derivative) are slow, the non-adiabatic couplings ($\langle \phi_a|\partial_\mu \phi_b \rangle$) are small, and a gap condition between the eigenvalue $\lambda_a$ and the other eigenvalues is satisfied.\\

If $\mathcal C$ passes through a point $x_* \in M$ where $\lambda_a$ and $\lambda_b$ cross, $\lambda_a(x_*) = \lambda_b(x_*)$, with a rapid passage in the neighbourhood of $x_*$, and with the adiabatic condition (\ref{adiabcond}) satisfied elsewhere; then we have
\begin{eqnarray}
\psi(T) & \simeq & e^{-\ihbar^{-1} \int_0^{t_*} \lambda_a(x(t))dt} e^{- \int_{\mathcal C:x(0) \to x_*} A_a(x)} \nonumber \\
& & \qquad \times e^{-\ihbar^{-1} \int_{t_*}^T \lambda_b(x(t))dt} e^{-\int_{\mathcal C:x_* \to x(T)} A_b(x)} \psi_b(x(T))
\end{eqnarray}
The rapid adiabatic passage method of quantum control \cite{Guerin} is based on this equation. To reach $\psi_{target} = \phi_b(x(T))$ it needs to find a path $\mathcal C$ passing through a crossing point of $\lambda_a$ and $\lambda_b$ (we can also pass by several crossing points with some intermediate eigenstates).

\subsection{Geometric approach}
As shown by Simon \cite{Simon}, the geometric phase of the adiabatic approximation is associated with a gauge theory described by a connection on a principal $U(1)$-bundle ($U(1)$ is the set of unit module complex numbers). $A_a \in \Omega^1 M$ plays the role of a gauge potential which defines a gauge field
\begin{equation}
F_a = dA_a = \langle d\phi_a| \wedge |d\phi_a \rangle \in \Omega^2 M
\end{equation}
$\wedge$ denotes the exterior product of differential forms. $F_a$ is called the adiabatic curvature.
\begin{prop}
The adiabatic curvature $F_a(x)$ is a measure at the point $x$ of the non-adiabaticity involving the state $\phi_a$.
\end{prop}
\begin{proof}
$\{\phi_b\}_b$ constitutes an orthonormal basis of $\S$. By using the closure relation we have
\begin{eqnarray}
F_a & = & \langle d\phi_a|\wedge|d\phi_a \rangle \\
& = & \sum_b \langle d\phi_a|\phi_b \rangle \wedge \langle \phi_b |d\phi_a \rangle \\
& = & - \sum_b \langle \phi_a|d\phi_b \rangle \wedge \langle \phi_b|d\phi_a \rangle
\end{eqnarray}
$\langle \phi_a |d\phi_a \rangle \wedge \langle \phi_a|d\phi_a \rangle = 0$ and moreover we have for $b\not=a$
\begin{equation}
H\phi_a = \lambda_a \phi_a  \Rightarrow \langle \phi_b|dH|\phi_a \rangle + \lambda_b \langle \phi_b|d\phi_a \rangle = \lambda_a\langle \phi_b|d\phi_a \rangle
\end{equation}
We have then
\begin{equation}
F_a = \sum_{b \not= a} \frac{\langle \phi_a|dH|\phi_b\rangle \wedge \langle \phi_b|dH|\phi_a \rangle}{(\lambda_a - \lambda_b)^2} 
\end{equation}
The adiabatic condition (\ref{adiabcond}) is then equivalent to $\hbar^2 |i_v F_a | \ll 1$ where $i$ is the interior product and $v = \dot x^\mu \frac{\partial}{\partial x^\mu}$ is the speed tangent vector of $\mathcal C$.
\end{proof}
$F_a$ diverges at the crossing points of $\lambda_a$ (which are the singularities of the Simon principal bundle) if the non-adiabatic couplings are different from zero. In the case where $\dim M = 2$, in place of drawing the eigenvalue surfaces like in \cite{Guerin} we can draw the field strength densities $F_{a \mu \nu}$ to locate the crossing points. The interests of studying the adiabatic curvature in place of the eigenvalue surfaces are: $F_a$ is zero at a crossing point where the adiabatic couplings are zero (such a crossing does not induce rapid transitions); $F_a$ shows the distribution of the non-adiabatic couplings around the crossings \cite{Viennot3}; and $F_a$ can be generalized to some non-hermitian cases where the eigenvalue surfaces are complex surfaces \cite{Viennot4}.\\

If now the eigenvalue is degenerate with the associated eigenvectors $\{\phi_a\}_{a \in I}$ ($I$ is a set of indexes), or if we consider a weaker adiabatic approximation consisting to assume that the wave function remains projected onto a group of several eigenvectors $\{\phi_a\}_{a \in I}$; then the gauge potential becomes
\begin{equation}
A_I \in \Omega^1(M,\mathfrak u(n)) \qquad A_{I,ab} = \langle \phi_a|d\phi_b \rangle \quad \forall a,b\in I
\end{equation}
where $\mathfrak u(n)$ is the set of anti-self-adjoint matrices of order $n$ (the number of elements in $I$) and $A_{I,ab}$ denotes the matricial element of $A_I$ at the row $a$ and the column $b$. The adiabatic approximation becomes (for a single degenerate eigenvalue)
\begin{equation}
\psi(T) \simeq  e^{-\ihbar^{-1} \int_0^T \lambda_a(x(t))dt} \sum_{b \in I} \left[\Pe^{-\oint_{\mathcal C} A_I(x)} \right]_{ba} \phi_b(x(T))
\end{equation}
where $\Pe$ denotes the path-ordered exponential, i.e. the Dyson series along a path:
\begin{equation}
\frac{d}{dt} \Pe^{- \int_{\mathcal C:x(0) \to x(t)} A_I(x)} = - A_{I \mu}(x(t)) \dot x^\mu(t) \Pe^{- \int_{\mathcal C:x(0) \to x(t)} A_I(x)}
\end{equation}
The ``non-abelian'' geometric phase is the Wilson loop $\Pe^{-\oint_{\mathcal C} A_I(x)} \in U(n)$ ($U(n)$ is the group of unitary matrices of order $n$) which is associated with a connection on a principal $U(n)$-bundle. In holonomic quantum computation \cite{Zanardi, Lucarelli}, the Wilson loops are used to perform quantum logic gates. The non-abelian adiabatic curvature is
\begin{equation}
F_I = dA_I + A_I \wedge A_I \in \Omega^2 (M,\mathfrak u(n))
\end{equation}
\begin{prop}
The non-abelian adiabatic curvature $F_I(x)$ is a measure at the point $x$ of the non-adiabaticity between the space spanned by $\{\phi_a\}_{a \in I}$ and its orthogonal supplement, but it is not sensitive to the non-adiabaticity inner the space spanned by $\{\phi_a\}_{a \in I}$.
\end{prop}
\begin{proof}
After some algebra similar to the non-degenerate case, we find
\begin{eqnarray}
F_{I,ab} & = & \left( \langle \partial_\mu \phi_a|\partial_\nu \phi_b \rangle + \sum_{c \in I} \langle \phi_a|\partial_\mu \phi_c \rangle \langle \phi_c|\partial_\nu \phi_b \rangle \right) dx^\mu \wedge dx^\nu \\
& = & \sum_{d \not\in I} \frac{\langle \phi_a|dH|\phi_d \rangle \wedge \langle \phi_d|dH|\phi_b \rangle}{(\lambda_a-\lambda_d)(\lambda_b-\lambda_d)}
\end{eqnarray}
\end{proof}
Remark: if the set of vectors is complete, i.e. $I=\{1,...,\dim \mathcal \S\}$, then $F_I = 0$.

\section{Adiabatic quantum control hampered by entanglement}
We consider now that the quantum system is in ``contact'' with another quantum ``object'' that we call the ``environment''. We call ``universe'' the composite system constituted by the quantum system and by the environment. We denote by $\E$ the Hilbert space of the environment and by $\SE$ the Hilbert space of the universe. The dynamics of the universe is governed by the self-adjoint Hamiltonian
\begin{equation}
H_{\mathcal U}(x) = H_{\mathcal S}(x) \otimes 1_\E + 1_\S \otimes H_{\mathcal E}(x) + V(x)
\end{equation}
where $H_{\mathcal S} \in \mathcal L(\S)$ and $H_{\mathcal E} \in \mathcal L(\E)$ are the Hamiltonians of the system and of the environment when they are separated, and $V(x) \in \mathcal L(\SE)$ is the coupling operator. Let $\psi(t) \in \SE$ be the solution of the Schr\"odinger equation of the universe. We are interested by the state ``reduced'' to the system which is represented by the density matrix
\begin{equation}
\rho_\psi(t) = \tr_\E |\psi(t) \rrangle \llangle \psi(t)|
\end{equation}
where $\llangle .|. \rrangle$ denotes the scalar product of $\SE$ and $\tr_\E$ denotes the partial trace on $\E$. If $\rho_\psi$ has a rank equal to one, there exists $\varphi \in \S$ such that $\rho_\psi = |\varphi \rangle \langle \varphi|$ ($\rho_\psi$ is said to be a pure state), and $\varphi$ is the single state which can be attributed to the system. If the rank of $\rho_\psi$ is larger that one, we cannot attribute a single state to the system ($\rho_\psi$ is said to be a mixed state), the system and the environment are entangled. If $V\not=0$ the dynamic transforms pure states into mixed states.\\
The role of the partial trace on $\E$ is to lose information concerning the environment. Indeed, the ``experimentalist'' controls only the system, and not directly the environment (even if this environment ``feels'' the control). The adiabatic regime is assumed for the system, not for the universe. We consider then a weaker adiabatic assumption consisting to assume an adiabatic evolution for the system but not necessarily for the environment. The $C^*$-geometric phases have been introduced in \cite{Viennot1} as a framework describing this situation.

\subsection{$C^*$-geometric phases}
Let $\{\lambda_a(x)\}_a$ be the instantaneous eigenvalues of the universe (we suppose that they are not degenerate for all $x \in M$ except for some isolated points) and $\{\phi_a\}$ be the associated eigenvectors.
\begin{equation}
H_{\mathcal U}(x) \phi_a(x) = \lambda_a(x) \phi_a(x)
\end{equation}
Following \cite{Viennot1} we can define a geometric phase with values in the $C^*$-algebra $\mathfrak s = \mathcal L(\S)$ as being
\begin{equation}
\aPe^{-\oint_{\mathcal C} \mathcal A_a(x)} \in \mathfrak s
\end{equation}
where $\aPe$ is the path anti-ordered exponential, i.e.
\begin{equation}
\frac{d}{dt} \aPe^{- \int_{\mathcal C:x(0)\to x(t)} \mathcal A_a(x)} = - \aPe^{- \int_{\mathcal C:x(0)\to x(t)} \mathcal A_a(x)} \mathcal A_{a\mu}(x(t)) \dot x^\mu(t)
\end{equation}
Let $\rho_a(x) = \tr_\E |\phi_a(x) \rrangle \llangle \phi_a(x)|$ be the density eigenmatrix. The generator of the $C^*$-geometric phase is defined by
\begin{equation}
\mathcal A_a = \tr_\E \left(|d\phi_a \rrangle \llangle \phi_a| \right) \rho_a^{-1} \in \Omega^1(M,\mathfrak s)
\end{equation}
where $\rho_a^{-1}$ is the pseudo-inverse of $\rho_a$, i.e. $\rho_a^{-1} \rho_a = 1_\S - P_{\ker \rho_a}$ ($P_{\ker \rho_a}$ is the orthogonal projector onto the kernel of $\rho_a$).

\subsection{Adiabatic fields}
In \cite{Viennot1} we have shown that the $C^*$-geometric phases are associated with an higher gauge theory (a connective structure on a 2-bundle \cite{Baez}) which is characterized by two fields:
\begin{itemize}
\item the adiabatic curving:
\begin{equation}
B_a = d\mathcal A_a - \mathcal A_a \wedge \mathcal A_a \in \Omega^2(M,\mathfrak s)
\end{equation}
\item the adiabatic fake curvature:
\begin{equation}
F_a = dA_a - A_a \wedge A_a - B_a \in \Omega^2(M,\mathfrak s)
\end{equation}
\end{itemize}
where the reduced potential is defined by
\begin{equation}
A_a = \tr_\E \left(P_a|d\phi_a \rrangle \llangle \phi_a| \right) \rho_a^{-1}
\end{equation}
where $P_a$ is the projection onto the eigensubspace associated with $\lambda_a$ which is considered as a ``non-commutative eigenvalue'' (see \cite{Viennot1}), usually we can expect that $P_a$ is simply $|\phi_a \rrangle \llangle \phi_a|$ ($A_a$ is then $P_a$ multiplied by the usual geometric phase generator of the universe). Since these fields are $\mathfrak s$-valued, they have statistical interpretations with respect to mixed states associated with the entanglement. More precisely, the physical meaning is not directly supported by these fields, but by their statistical averages:
\begin{equation}
\tr_\S (\rho_a B_a) \in \Omega^2 M
\end{equation}
\begin{equation}
\tr_\S (\rho_a F_a) \in \Omega^2 M
\end{equation}
Since the entanglement of the quantum system with the environment is responsible for a lost of information in the partial trace $\tr_\E$, it is interesting to consider also the von Neuman entropy of the density eigenmatrix:
\begin{equation}
- \tr_\S (\rho_a \ln \rho_a) \in \Omega^0 M
\end{equation}
which can be viewed as a measure of the information lack for the system where its mixed state is described by the density eigenmatrix $\rho_a$.\\

By construction, the average adiabatic fake curvature seems to have the same interpretation than the usual adiabatic curvature of isolated systems. It measures the local non-adiabaticity. This is well the fake curvature which must be considered and not $dA_a - A_a \wedge A_a$ (the correction by $B_a$ is necessary). This is induced by the mathematical structure of an higher gauge theory (see \cite{Baez}) but with a more pragmatic approach we will justify this fact by the examples which follow in the rest of this paper.\\

The role of the adiabatic curving is enlighten by the following property.
\begin{prop}
The average adiabatic curving $\tr_\S(\rho_a(x) B_a(x))$ is a measure of the entropy variation associated with $\phi_a$ and which is induced by variations of the control parameters in the neighbourhood of $x$.
\end{prop}

\begin{proof}
Let $\mathcal C_x$ be an infinitesimal closed loop in $M$, starting and ending at $x \in M$. Let $\mathcal S_x$ be a surface in $M$ having $\mathcal C_x$ as boundary. We denote by $\Delta$ the area of $\mathcal S_x$ which is in the neighbourhood of zero. Let $e^{\iint_{\mathcal S_x} B_a} \rho_a(x) \simeq e^{\oint_{\mathcal C_x} \mathcal A_a} \rho_a(x)$ be the density matrix obtained by the ``parallel transport'' of $\rho_a(x)$ along $\mathcal C_x$ (since the loop $\mathcal C_x$ is infinitesimal the path-ordered exponential is approximately equal to the matrix exponential, and the Wilson loop $e^{\oint_{\mathcal C_x} \mathcal A_a}$ is approximately equal to $e^{\iint_{\mathcal S_x} B_a}$ by a Stokes theorem). By using the Baker-Campbell-Hausdorff formula \cite{BCH}, we have
\begin{eqnarray}
\ln \left(e^{\iint_{\mathcal S_x} B_a} \rho_a \right) & = & \iint_{\mathcal S_x} B_a + \ln \rho_a + \frac{1}{2} \left[ \iint_{\mathcal S_x} B_a, \ln \rho_a \right] \nonumber \\
& & + \frac{1}{12} \left[ \iint_{\mathcal S_x} B_a, \left[ \iint_{\mathcal S_x} B_a, \ln \rho_a \right] \right] \nonumber \\
& & - \frac{1}{12} \left[ \ln \rho_a, \left[ \iint_{\mathcal S_x} B_a, \ln \rho_a \right] \right] \nonumber \\
& & + ...
\end{eqnarray}
We have then
\begin{eqnarray}
\tr_\S \left(\rho_a \ln \left(e^{\iint_{\mathcal S_x} B_a} \rho_a \right) \right)  & =  & \tr_\S \left(\rho_a \iint_{\mathcal S_x} B_a \right) \nonumber \\
& & + \tr_\S \left( \rho_a \ln \rho_a \right) + \mathcal O(\Delta^2)
\end{eqnarray}
because of the cyclicity of the trace we have $\tr(\rho[B,\rho])=\tr(\rho B\ln \rho)-\tr(\rho \ln \rho B) = \tr(\ln \rho \rho B) - \tr(\rho \ln \rho B) = 0$ ($\rho \ln \rho = \ln \rho \rho)$). We have similar calculations for higher orders. Finally we have
\begin{eqnarray}
& & \tr_\S \left(\rho_a \iint_{\mathcal S_x} B_a \right) \nonumber \\
& & \qquad = - \tr_\S \left(\rho_a \left(\ln \rho_a - \ln \left(e^{\iint_{\mathcal S_x} B_a} \rho_a \right) \right) \right) + \mathcal O(\Delta^2) \\
& & \qquad = S\left(\rho_a \left\| e^{\iint_{\mathcal S_x} B_a} \rho_a \right. \right) + \mathcal O(\Delta^2)
\end{eqnarray}
where $S(\rho\|\tau) = -\tr(\rho(\ln \rho - \ln \tau))$ is the relative entropy (see \cite{Bengtsson}). $\tr_\S \left(\rho_a \iint_{\mathcal S_x} B_a \right)$ is then the relative entropy of $\rho_a(x)$ with respect to its parallel transport along an infinitesimal loop passing through $x$. By writing $\iint_{\mathcal S_x} B_a = B_{a;12} \Delta + \mathcal O(\Delta^2)$ (the indices $12$ being associated with local coordinates along $\mathcal S_x$) we see that $\tr_\S (\rho_a(x) B_a(x))$ is a measure of the entropy variation induced by the transport of $\rho_a$ in the neighbourhood of $x$.
\end{proof}
The increase of the entropy is associated with an increase of the entanglement between the system and the environment (and dynamically it is associated with decoherence processes). In quantum control, we can define the decoherence as a dynamical process associated with an increase of the entropy and of the entanglement. We have two kinds of decoherence, a ``\textit{local decoherence}'' associated with $- \tr_\S (\rho_a(x) \ln \rho_a(x))$ (decoherence induced by the point $x$) and a ``\textit{kinematic decoherence}'' associated with $\tr_\S (\rho_a(x) B_a(x))$ (decoherence induced by loops passing through $x$). 

\subsection{Two reference cases}
In order to illustrate the roles of the curving and of the fake curvature, and to enlighten their interpretations, we consider two simple cases where these fields can be expressed by using the curvatures of the system and of the environment.

\subsubsection{Factorizable eigenstate:}
\label{factorizable}
We suppose that an eigenvector of the universe is $\phi_a(x) = \zeta_i(x) \otimes \xi_\alpha(x)$ where $\zeta_i(x)$ is an eigenvector of $H_{\mathcal S}(x)$ (associated with a non-degenerate eigenvalue) and $\xi_\alpha(x)$ is an eigenvector of $H_{\mathcal E}(x)$ (associated with a non-degenerate eigenvalue). We do not need to suppose that the other eigenvectors of the universe have a same decomposition. In that case, the density eigenmatrix is the projection (pure state) $\rho_a = |\zeta_i \rangle \langle \zeta_i| = P_{\zeta_i}$ and $P_a = P_{\zeta_i} \otimes P_{\xi_\alpha}$ ($P_{\xi_\alpha} = |\xi_\alpha \rangle \langle \xi_\alpha|$). Since the density matrix is a pure state, its von Neuman entropy is zero and no local decoherence associated with an adiabatic approximation involving only $\phi_a$ occurs. The gauge potential and the reduced potential are
\begin{eqnarray}
\mathcal A_a & = & \tilde {\mathcal A}_i + A_{\mathcal E,\alpha} P_{\zeta_i} \\
A_a & = & (A_{\mathcal S,i} + A_{\mathcal E,\alpha}) P_{\zeta_i}
\end{eqnarray}
where $\tilde {\mathcal A}_i = |d\zeta_i \rangle \langle \zeta_i| \in \Omega^1(M,\mathfrak s)$ is the $C^*$-geometric phase generator for the system without environment, $A_{\mathcal S,i} = \langle \zeta_i|d\zeta_i\rangle \in \Omega^1 M$ is the (usual) geometric phase generator for the isolated system, and $A_{\mathcal E,\alpha} = \langle \xi_\alpha |d\xi_\alpha \rangle \in \Omega^1 M$ is the (usual) geometric phase generator for the isolated environment.\\
The curving and the fake curvature are
\begin{eqnarray}
B_a  & = & \tilde B_a + F_{\mathcal E,\alpha} P_{\zeta_i} - A_{\mathcal E,\alpha} \wedge P_{\zeta_i}(\tilde {\mathcal A}_i + \tilde {\mathcal A}_i^\dagger) \\
F_a & = & (F_{\mathcal S,i} + F_{\mathcal E,\alpha}) P_{\zeta_i} - (A_{\mathcal S,i} + A_{\mathcal E,\alpha}) \wedge (\mathcal A_a + \mathcal A_a^\dagger) - B_a
\end{eqnarray}
where $F_{\mathcal S,i} = dA_{\mathcal S,i} \in \Omega^2 M$ is the adiabatic curvature of the isolated system, $F_{\mathcal E,\alpha} = dA_{\mathcal E,\alpha} \in \Omega^2M$ is the adiabatic curvature of the isolated environment, and $\tilde B_a = d\tilde {\mathcal A}_i - \tilde {\mathcal A}_i \wedge \tilde {\mathcal A}_i = - |d\zeta_i \rangle \wedge \langle d\zeta_i| + A_{\mathcal S,i} \wedge \mathcal A_a \in \Omega^2(M,\mathfrak s)$ is the curving associated with the $C^*$-geometric phase generator of the system without environment. These expressions seem to contain complicated terms, but in fact the averages are very simple:
\begin{eqnarray}
\tr_\S (\rho_a F_a) & = & F_{\mathcal S,i} \\
\tr_\E (\rho_a B_a) & = & F_{\mathcal E,\alpha}
\end{eqnarray}
The average fake curvature is the curvature of the isolated system, in accordance with their common interpretation. Because of the system and the environment are not entangled, if the universe is in the state $\phi_a$, the non-adiabatic processes are the same for the system in contact with the environment and for the isolated system.\\
The average curving is the curvature of the environment. It is this curvature which measures the kinematic decoherence processes. The explanation of this fact is the following. We assume an adiabatic approximation for the system in contact with the environment, but we do not assume that the dynamics of the environment is adiabatic. Indeed if we assume a total adiabaticity (system and environment), the evolution of the universe is ``strongly'' adiabatic and is characterized by the universe geometric phase generator $\llangle \phi_a|d\phi_a \rrangle$. This is not in accordance with the problem of quantum control of a system in contact with an environment. We directly control only the system (and we can only assume the adiabaticity for the system). The environment feels the control, but the ``experimentalist'' does not know the environment and its dynamics. This is the sense of the partial trace $\tr_\E$, the information concerning the environment is lost. The dynamics of the universe under the control is then ``weakly'' adiabatic and is characterized by the $C^*$-geometric phase generator $|d\phi_a \rrangle \llangle \phi_a| \rho_a^{-1}$. If the universe is in the state $\zeta_i \otimes \xi_\alpha$, and if the control path $\mathcal C$ passes through a region of $M$ with a strong curvature of the environment, then non-adiabatic transitions occur from $\xi_\alpha$ to another state $\xi_\beta$. But $\zeta_i \otimes \xi_\beta$ is not necessarily an eigenvector of the universe (we have not suppose that all eigenstates of the universe are factorizable). $\zeta_i \otimes \xi_\beta$ can be a superposition of eigenstates of the universe, and the dynamics will induce Rabi oscillations between these states. These oscillations will destroy the factorization, and the system and the environment will become entangled.

\subsubsection{Eigenstate as Schmidt decomposition:}
\label{schmidt}
We suppose that an eigenvector of the universe has a Schmidt decomposition:
\begin{equation}
\phi_a(x) = \sum_{i \in I} \sqrt{p_i} \zeta_i(x) \otimes \xi_i(x) \qquad \sum_{i\in I} p_i = 1
\end{equation}
where $I$ is a subset of $\{1,...,\dim \S\}$ and $0<p_i<1$ are occupation probabilities independent of $x$. $\{\zeta_i\}_i$ and $\{\xi_i \}_i$ are eigenvectors of the system and of the environment (associated with non-degenerate eigenvalues). The density eigenmatrix is $\rho_a = \sum_{i \in I} p_i |\zeta_i \rangle \langle \zeta_i|$ and $P_a = \sum_{i \in I} P_{\zeta_i} \otimes P_{\xi_i}$. The von Neuman entropy is the same on the whole of $M$, $- \tr_\S(\rho_a \ln \rho_a) = -\sum_{i \in I} p_i \ln p_i$.\\
The gauge potential and the reduced potential are
\begin{eqnarray}
\mathcal A_a & = & \sum_{i\in I} \tilde A_i + \sum_{i,j\in I} \hat A_{\mathcal E,I,ij} |\zeta_i \rangle \langle \zeta_j| \\
A_a & = & \sum_{i \in I} (A_{\mathcal S,I,ii} + A_{\mathcal E,I,ii}) P_{\zeta_i}
\end{eqnarray}
where $\tilde {\mathcal A_i} = |d\zeta_i \rangle \langle \zeta_i| \in \Omega^1(M,\mathfrak s)$ is the $C^*$-geometric phase generator for the system without environment, $A_{\mathcal S,I} \in \Omega^1(M,\mathfrak u(n))$ ($A_{\mathcal S,I,ij} = \langle \zeta_i|d\zeta_j \rangle$ and $n = \card I$) is the non-abelian geometric phase generator for the isolated system, $A_{\mathcal E,I} \in \Omega^1(M,\mathfrak u(n))$ ($A_{\mathcal E,I,ij} = \langle \xi_i|d\xi_j \rangle$) is the non-abelian geometric phase generator for the isolated environment, and $\hat A_{\mathcal E,I} = \varpi^{-1} A_{\mathcal E,I} \varpi$ with $\varpi_{ij} = \sqrt{p_i} \delta_{ij}$.\\
The average fake curvature and the average curving are
\begin{eqnarray}
\tr_\S (\rho_a F_a) & = & \tr_{\mathbb C^n} \left( \varpi^2 ( F_{\mathcal S,I} + [A_{\mathcal S,I},\hat A_{\mathcal E,I}^t]) \right) \\
\tr_\S (\rho_a B_a) & = & \tr_{\mathbb C^n} \left(\varpi^2 (F_{\mathcal E,I} - A_{\mathcal S,I} \wedge A_{\mathcal S,I}) \right)
\end{eqnarray}
where $F_{\mathcal S,I} = dA_{\mathcal S,I} + A_{\mathcal S,I} \wedge A_{\mathcal S,I} \in \Omega^2(M,\mathfrak u(n))$ is the non-abelian curvature of the isolated system and $F_{\mathcal E,I} = dA_{\mathcal E,I} + A_{\mathcal E,I} \wedge A_{\mathcal E,I} \in \Omega^2(M,\mathfrak u(n))$ is the non-abelian curvature of the isolated environment ($A^t$ denotes the transposition of the matrix $A$).\\
The average fake curvature is then essentially the statistical average of the non-abelian curvature of the isolated system. In a same way, the average curving is essentially the statistical average of the non-abelian curvature of the isolated environment. The interpretations are then the same that for the previous example, but with an average associated with the superposition of factorized states of the Schmidt decomposition. The additional term in the average curving $\tr(\varpi^2 A_{\mathcal S,I} \wedge A_{\mathcal S,I}) = \sum_{i,j \in I} p_i A_{\mathcal S,I,ij} \wedge A_{\mathcal S,I,ji}$ characterizes the non-adiabatic transitions of the system inner the space spanned by $\{\zeta_i\}_{i \in I}$. These transitions can modify the superposition coefficients and induce Rabi oscillations between $\phi_a$ and others eigenvectors of the universe (and induce modifications of the entanglement). The additional term in the average fake curvature $\tr(\varpi^2 [A_{\mathcal S,I},\hat A_{\mathcal E,I}^t])$ characterizes non-adiabatic transitions for the system induced by its entanglement with the environment.\\
Remark: If the probabilities $p_i$ depend on $x$, we have a new gauge potential $\breve {\mathcal A}_a = \mathcal A_a + \sum_i d\ln \sqrt{p_i} P_{\zeta_i} = \mathcal A_a + \sum_i (d\varpi \varpi^{-1})_{ii} P_{\zeta_i}$. $\varpi$ plays the role of an usual gauge change and the results for the average fake curvature and for the average curving are similar to the case where the probabilities are independent of $x$.

\section{Example: STIRAP}
In this section, we illustrate the role of the fields associated with the higher gauge theory with a concrete example. We want also show that a geometric representation of the field strengths can be used to interpret the hampering of the quantum control induced by the entanglement of the system with its environment. Then it can be used to analyse the robustness of a control solution found by considering solely the system. We have chosen a very simple quantum control problem in order to avoid unnecessary complications which could hide the fundamental behaviours of the universe.\\
We consider a three level atom controlled by two laser Gaussian pulses. The first one, called ``pump'' pulse, is quasi-resonant with the transition $|1 \rangle \to |2 \rangle$; and the second one, called ``Stokes'' pulse, is quasi-resonant with the transition $|2 \rangle \to |3 \rangle$ ($(|i \rangle)_{i=1,2,3}$ being the atomic bare states). A second three level atom interacts with the first one and feels the laser pulses but with attenuated intensities (we set the attenuation as being a division by a factor $2$). This second atom constitutes the environment for the controlled atom. In the rotating wave approximation (see \cite{Guerin}), the Hamiltonian of the universe is
\begin{equation}
H_{\mathcal U}(x) = H_{\mathcal S}(x) \otimes 1_\E + 1_\S \otimes H_{\mathcal E}(x) + V(x)
\end{equation}
with in the basis $(|i\rangle)_{i=1,2,3}$ for the first atom
\begin{equation}
H_{\mathcal S}(x) = \frac{\hbar}{2} \left( \begin{array}{ccc} 0 & \Omega_P & 0 \\ \Omega_P & 2 \Delta_P & \Omega_S \\ 0 & \Omega_S & 2(\Delta_P-\Delta_S) \end{array} \right)
\end{equation}
and in the basis $(|i\rangle)_{i=1,2,3}$ for the second atom
\begin{equation}
H_{\mathcal E}(x) = \frac{\hbar}{2} \left( \begin{array}{ccc} 0 & \frac{1}{2} \Omega_P & 0 \\  \frac{1}{2} \Omega_P & 2 \Delta_P &  \frac{1}{2} \Omega_S \\ 0 &  \frac{1}{2} \Omega_S & 2(\Delta_P-\Delta_S) \end{array} \right)
\end{equation}
$\Omega_P = |\langle 1| \vec \mu \cdot \vec {\mathcal E}_P|2 \rangle|$ where $\vec \mu$ is the electric dipole moment of an atom and $\vec {\mathcal E}_P$ is the electric field of the pump pulse, $\Omega_S = |\langle 2| \vec \mu \cdot \vec {\mathcal E}_S|3 \rangle|$ where $\vec {\mathcal E}_S$ is the electric field of the Stokes pulse, $\hbar \Delta_P = (\epsilon_2 - \epsilon_1) - \hbar \omega_P$ and $\hbar \Delta_S = (\epsilon_3- \epsilon_2) - \hbar \omega_S$ where $(\epsilon_i)_{i=1,2,3}$ are the atomic bare energies and $\omega_P$ and $\omega_S$ are the laser frequencies. The laser frequencies (and then the detuning $\Delta_{P/S}$) are fixed; and $x = (\Omega_P,\Omega_S)$ constitute the control parameters. The control manifold is $M = \mathbb R^+ \times \mathbb R^+$.\\
For the sake of simplicity, we choose a simple operator $V$ to model the coupling between the two atoms. We consider two cases:
\begin{itemize}
\item a static coupling:
\begin{equation}
V = g \left( |2,3 \rrangle \llangle 3,2| + |3,2 \rrangle \llangle 2,3| \right)
\end{equation}
where $|i,j \rrangle = |i \rangle \otimes |j \rangle \in \SE$;
\item a dynamical coupling:
\begin{eqnarray}
V(x) & = & g \left( |\zeta_2(x) \otimes \xi_3(x) \rrangle \llangle \zeta_3(x) \otimes \xi_2(x)| \right. \nonumber \\
& & \quad \left. + |\zeta_3(x) \otimes \xi_2(x) \rrangle \llangle \zeta_2(x) \otimes \xi_3(x)| \right)
\end{eqnarray}
where $\{\zeta_i(x)\}_{i=1,2,3}$ are the eigenvectors of $H_{\mathcal S}(x)$ (continuous with respect to $x$) such that $|\zeta_i(0) \rangle = |i \rangle$; and $\{\xi_i(x)\}_{i=1,2,3}$ are the eigenvectors of $H_{\mathcal E}(x)$ (continuous with respect to $x$) such that $|\xi_i(0) \rangle = |i \rangle$ ($x=0$ is the point corresponding to off lasers, $\Omega_S = \Omega_P = 0$).
\end{itemize}
$g \geq 0$ is the coupling strength. We consider only perturbative couplings between the system and the environment ($g \ll 1$).\\

Let $\{\phi_a(x)\}_{a=1,...,9}$ be the eigenvectors of $H_{\mathcal U}(x)$, continuous with respect to $x$ and such that
\begin{eqnarray}
& & \lim_{g \to 0} \phi_1 = \zeta_1 \otimes \xi_1 \qquad \lim_{g \to 0} \phi_2 = \zeta_2 \otimes \xi_1 \qquad \lim_{g \to 0} \phi_3 = \zeta_3 \otimes \xi_1 \\
& & \lim_{g \to 0} \phi_4 = \zeta_1 \otimes \xi_2 \qquad \lim_{g \to 0} \phi_5 = \zeta_2 \otimes \xi_2 \qquad \lim_{g \to 0} \phi_6 = \zeta_3 \otimes \xi_2 \\
& & \lim_{g \to 0} \phi_7 = \zeta_1 \otimes \xi_3 \qquad \lim_{g \to 0} \phi_8 = \zeta_2 \otimes \xi_3 \qquad \lim_{g \to 0} \phi_9 = \zeta_3 \otimes \xi_3
\end{eqnarray}

The quantum control problem consists to reach the pure target state $\rho_{target} = |3 \rangle \langle 3|$ with the system initially in the pure state $\rho_0 = |1 \rangle \langle 1|$. In a first time, we recall the classical solution of the problem (the STIRAP solution) when the controlled atom is alone. In a second time, we study the robustness of this solution where the controlled atom is in contact with the second one, without coupling, with the static coupling and with the dynamical coupling.\\
The results of the control are computed by numerical integrations of the Schr\"odinger equation of the universe based on a second order differential scheme (see for example \cite{Leforestier}). The different field strengths ($F_{\mathcal S}$, $\tr_\S(\rho_a F_a)$, $\tr_\S(\rho_a B_a)$) are numerically computed by using methods issued from lattice gauge theory \cite{Kogut, Larsson} after a triangulation of the control manifold $M$ with a sufficiently thin triangular lattice.

\subsection{A single isolated atom}
The STIRAP solution \cite{Guerin} consists in the path $\mathcal C$ on $M$ defined by
\begin{equation}
t \mapsto x(t) = \left( \Omega_0 e^{-\frac{(t-t_P)^2}{\tau_P^2}}, \Omega_0 e^{-\frac{(t-t_S)^{2}}{\tau_S^2}} \right)
\end{equation}
with $\Omega_0 = 3.5\ ua$, $t_P=70\ ua$, $t_S= 20\ ua$, $\tau_P=\tau_S = 30\ ua$ (with $t_0 = -60\ ua$ and $T=140\ ua$) ($ua$= atomic units). This solution is counter-intuitive since it consists to start the Stokes pulse (which is quasi-resonant with the transition $|2 \rangle \to |3 \rangle$) before the pump pulse (which is quasi-resonant with the transition $|1 \rangle \to |2 \rangle$). To enlighten this control solution, we numerically integrate the Schr\"odinger equation for the system alone, and we consider the density matrix $\rho_{\psi}(t) = |\psi(t) \rangle \langle \psi(t)|$ with $\psi(t)$ the solution of the Schr\"odinger equation. The occupation probabilities of the bare states $\rho_{\psi,ii}(t) = \langle i|\rho_{\psi}(t)|i \rangle$ and the occupation probabilities of the instantaneous eigenstates $\tr_\S(\rho_\psi(t) P_i(x(t)))$ with $P_i(x) = |\zeta_i(x) \rangle \langle \zeta_i(x)|$ are shown figure \ref{occ_prob_alone}.
\begin{figure}
\begin{center}
\includegraphics[width=12.8cm]{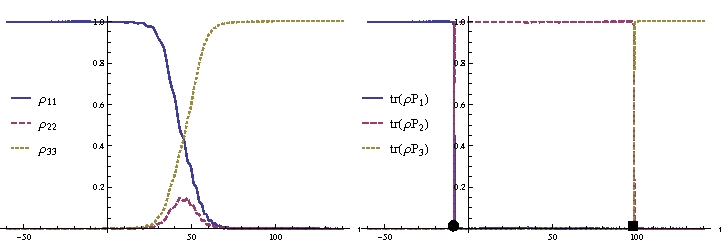}
\caption{\label{occ_prob_alone} Left: occupation probabilities of the bare states $\rho_{\psi,ii}$ with respect to $t$. Right: occupation probabilities of the instantaneous eigenstates $\tr_\S(\rho_\psi P_i)$ with respect to $t$. The event $\bullet$ is the transition $\zeta_1 \to \zeta_2$ and the event $\blacksquare$ is the transition $\zeta_2 \to \zeta_3$.}
\end{center}
\end{figure}
We can interpret the solution by using the adiabatic curvatures $F_{\mathcal S,i} = dA_{\mathcal S,i}$ with $A_{\mathcal S,i} = \langle \zeta_i|d\zeta_i \rangle$. Figure \ref{fields_alone} shows the densities of the field strengths $F_{\mathcal S,i,12} = \langle \frac{\partial \zeta_i}{\partial x^1} | \wedge | \frac{\partial \zeta_i}{\partial x^2} \rangle$.
\begin{figure}
\begin{center}
\includegraphics[width=15.5cm]{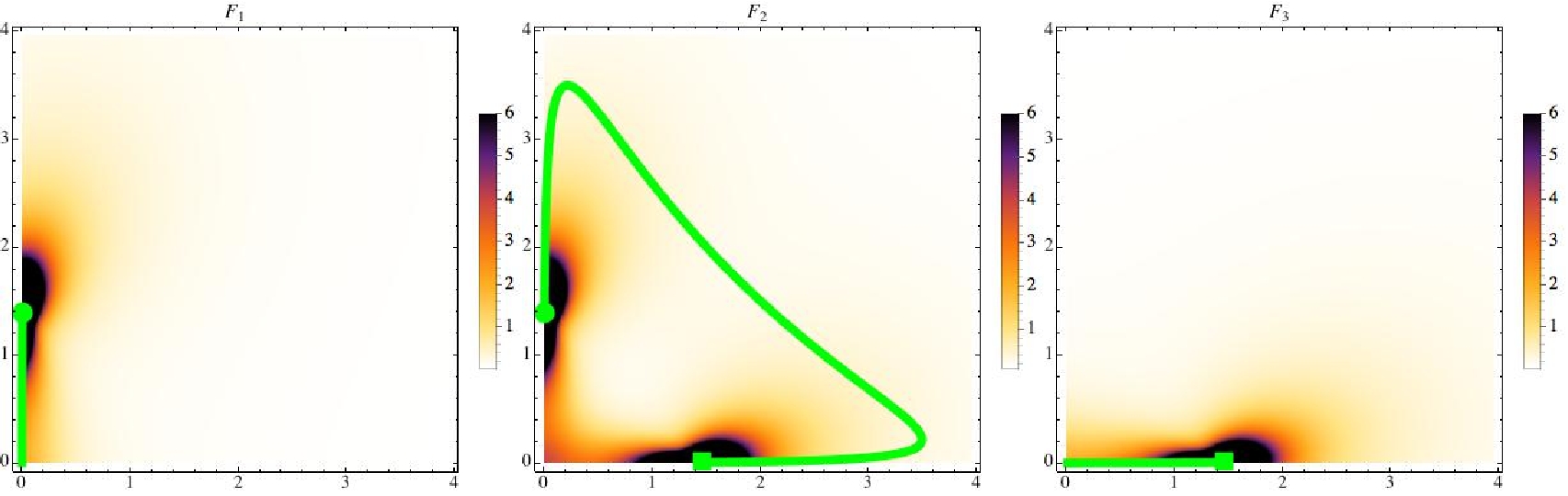} 
\caption{\label{fields_alone} Densities of the adiabatic curvatures $F_{\mathcal S,i}$ in $M$, with the path $\mathcal C$ and the events $\bullet$ (transition $\zeta_1 \to \zeta_2$) and $\blacksquare$ (transition $\zeta_2 \to \zeta_3$).}
\end{center}
\end{figure}
Starting from $P_1$, the dynamics passes to $P_2$ at $\bullet$ which corresponds to the singularity common to $F_{\mathcal S,1}$ and $F_{\mathcal S,2}$ (the crossing of the two associated eigenvalues). The dynamics passes from $P_2$ to $P_3$ at $\blacksquare$ the singularity common to $F_{\mathcal S,2}$ and $F_{\mathcal S,3}$. This is in accordance with the adiabatic quantum control method based on rapid adiabatic passages (see also \cite{Guerin}).\\

In the sequel, we do not change the path $\mathcal C$ (the STIRAP solution of the control problem), but we study the robustness of this solution with the entanglement of the controlled atom with the atom constituting the environment.

\subsection{Two atoms without coupling}
We begin by studying the case without coupling between the two atoms ($g=0$). The eigenvectors $\phi_a$ are then factorizable, and the rank of the density eigenmatrices $\rho_a = \tr_\E |\phi_a \rrangle \llangle \phi_a|$ is equal to one. The von Neuman entropy $- \tr_\S (\rho_a \ln \rho_a)$ is then zero for all states.\\
We integrate the Schr\"odinger equation of the universe with two initial conditions, the first one is without state superposition $\psi(t_0) = \phi_1(0)$ and the second one with initial environment state superposition $\psi(t_0) = \frac{1}{\sqrt{2}} (\phi_1(0) + \phi_7(0))$. In the two cases we have $\rho_\psi(t_0) = \tr_\S |\psi(0) \rrangle \llangle \psi(0)| = |1 \rangle \langle 1|$. The occupation probabilities of the bare states $\rho_{\psi,ii}$, the occupation probabilities of the instantaneous eigenvectors of the system $\tr_\S (\rho_\psi P_i)$ ($P_i = |\zeta_i \rangle \langle \zeta_i|$) and the occupation probabilities of the eigenvectors of the universe $|\llangle \phi_a|\psi \rrangle|^2$ are shown figure \ref{occ_prob_zero}.
\begin{figure}
\begin{center}
\includegraphics[width=12.8cm]{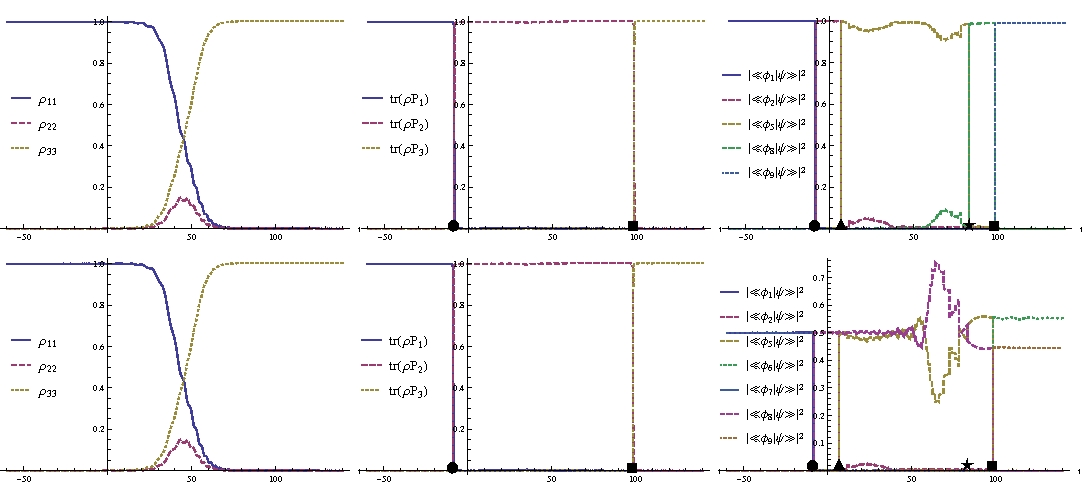}
\caption{\label{occ_prob_zero} Left: occupation probabilities of the bare states $\rho_{\psi,ii}$ with respect to $t$. Middle: occupation probabilities of the instantenous eigenstates of the system $\tr_\S(\rho_\psi P_i)$ with respect to $t$. Right: occupation probabilities of the instantenous eigenstates of the universe $|\llangle \phi_a|\psi \rrangle|^2$ with respect to $t$. Up: for the initial condition $\psi(t_0) = \phi_1(0)$. Down: for the initial condition $\psi(t_0) = \frac{1}{\sqrt{2}} (\phi_1(0) + \phi_7(0))$. The event $\bullet$ is the transition $\zeta_1 \to \zeta_2$ ($\phi_1 \to \phi_2$) and the event $\blacksquare$ is the transition $\zeta_2 \to \zeta_3$ ($\phi_5 \to \phi_9$). The event $\blacktriangle$ is the transition $\phi_2 \to \phi_5$ and the event $\star$ is the transition $\phi_5 \to \phi_8$.}
\end{center}
\end{figure}
In the two cases, we see that the adiabaticity is satisfied for the system whereas this is not completely the case for the universe. But since all states of the universe are factorizable, no decoherence significantly occurs on the controlled dynamics (for the first case the entropy remains equal to zero and it increases to a very small value in the second case as shown figure \ref{entropy_zero}).
\begin{figure}
\begin{center}
\includegraphics[width=6.4cm]{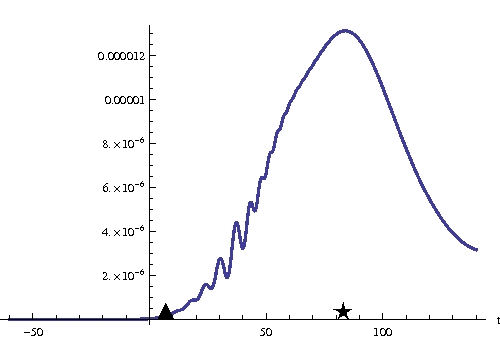}
\caption{\label{entropy_zero} von Neuman entropy of the density matrix $- \tr_\S (\rho_\psi \ln \rho_\psi)$ with respect to $t$ for the initial condition $\psi(0) = \frac{1}{\sqrt{2}} (\phi_1(0) + \phi_7(0))$.}
\end{center}
\end{figure}
We can interpret the results by drawing the densities in $M$ of the average adiabatic fake curvatures $\tr_\S(\rho_a F_a)$ (figure \ref{fake_zero}) and of average adiabatic curving $\tr_\S(\rho_a B_a)$ (figure \ref{curving_zero}).
\begin{figure}
\begin{center}
\includegraphics[width=15.5cm]{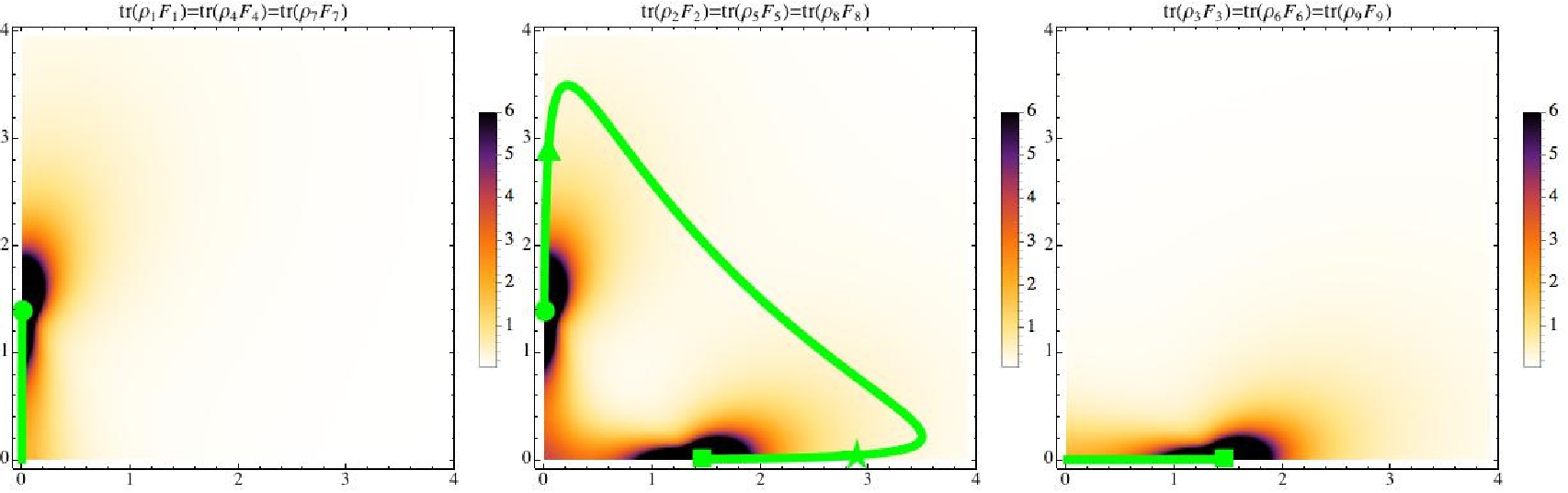}
\caption{\label{fake_zero} Densities of the average fake curvature $\tr_\S(\rho_a F_a)$ in $M$, with the path $\mathcal C$ and the events $\bullet$ (transition $\zeta_1 \to \zeta_2$), $\blacktriangle$ (transition $\phi_2 \to \phi_5$), $\star$ (transition $\phi_5 \to \phi_8$) and $\blacksquare$ (transition $\zeta_2 \to \zeta_3$).}
\end{center}
\end{figure}
\begin{figure}
\begin{center}
\includegraphics[width=15.5cm]{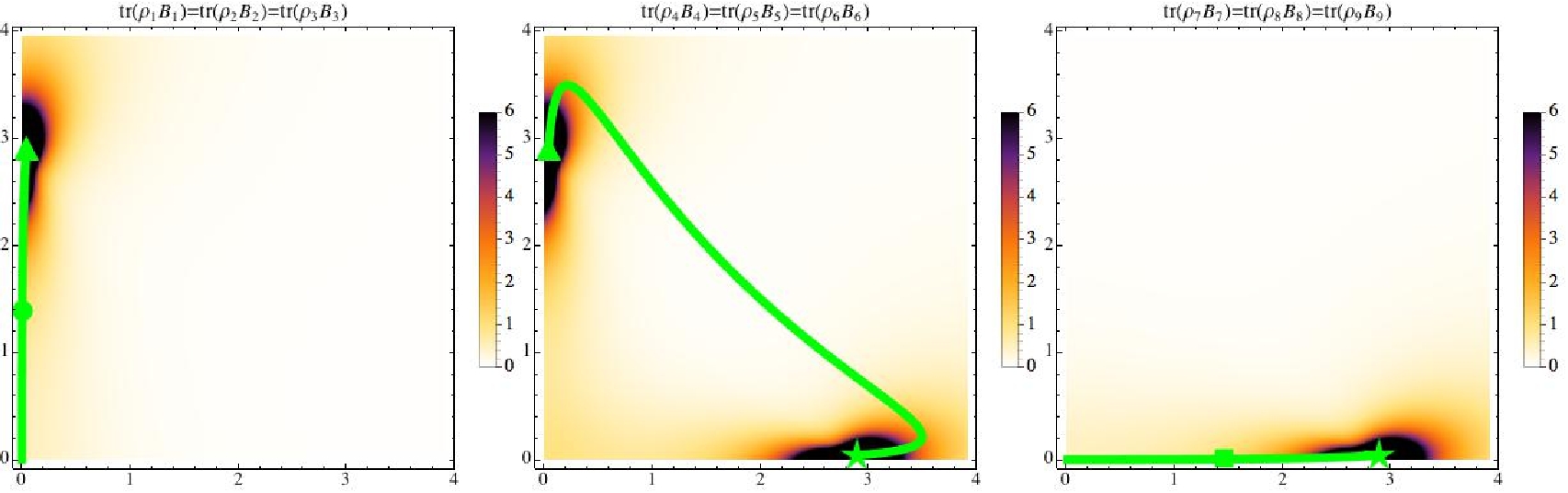}
\caption{\label{curving_zero} Densities of the average curving $\tr_\S(\rho_a B_a)$ in $M$, with the path $\mathcal C$ and the events $\bullet$ (transition $\zeta_1 \to \zeta_2$), $\blacktriangle$ (transition $\phi_2 \to \phi_5$), $\star$ (transition $\phi_5 \to \phi_8$) and $\blacksquare$ (transition $\zeta_2 \to \zeta_3$).}
\end{center}
\end{figure}
As shown in section \ref{factorizable} the average fake curvature is equal to the curvature of the system and the average curving is equal to the curvature of the environment. To summarize we have:
\begin{itemize}
\item[$\bullet$] the path $\mathcal C$ passes through a singularity of $\tr_\S(\rho_1 F_1)$ and $\tr_\S(\rho_2 F_2)$ which induces a transition $P_1 \to P_2$ for the system (and a transition $\phi_1 \to \phi_2$ for the universe);
\item[$\blacktriangle$] $\mathcal C$ passes through a singularity of $\tr_\S(\rho_2 B_2)$ and $\tr_\S(\rho_5 B_5)$ which induces a transition $\phi_2 \to \phi_5$ (without non-adiabatic effect for the system);
\item[$\star$] $\mathcal C$ passes through a singularity of $\tr_\S(\rho_5 B_5)$ and $\tr_S(\rho_8 B_8)$ which induces a transition $\phi_5 \to \phi_8$ (without non-adiabatic effect for the system);
\item[$\blacksquare$] $\mathcal C$ passes through a singularity of $\tr_\S(\rho_8 F_8)$ and $\tr_\S(\rho_9 F_9)$ which induces a transition $P_2 \to P_3$ for the system (and a transition $\phi_8 \to \phi_9$ for the universe).
\end{itemize}
The passages by the singularities of the curving have not significant consequences on the control in this case (except the very small increase of the entropy associated with the non-adiabatic transitions in environment for the case with a superposition of eigenstates).

\subsection{Two atoms with a static coupling}
We remake the previous study with a static coupling ($g=0.1\ ua$). In this case, all the density eigenmatrices $\rho_a$ are invertible (their rank is equal to $3$) for non zero laser intensities. We consider three initial conditions $\psi(t_0) = \phi_1(0)$, $\psi(t_0) = \phi_4(0)$ and $\psi(t_0) = \phi_7(0)$. The occupation probabilities are drawn figures \ref{occ_prob_static1} and the entropies of the density matrix are drawn figure \ref{entropy_static}.
\begin{figure}
\begin{center}
\includegraphics[width=12.8cm]{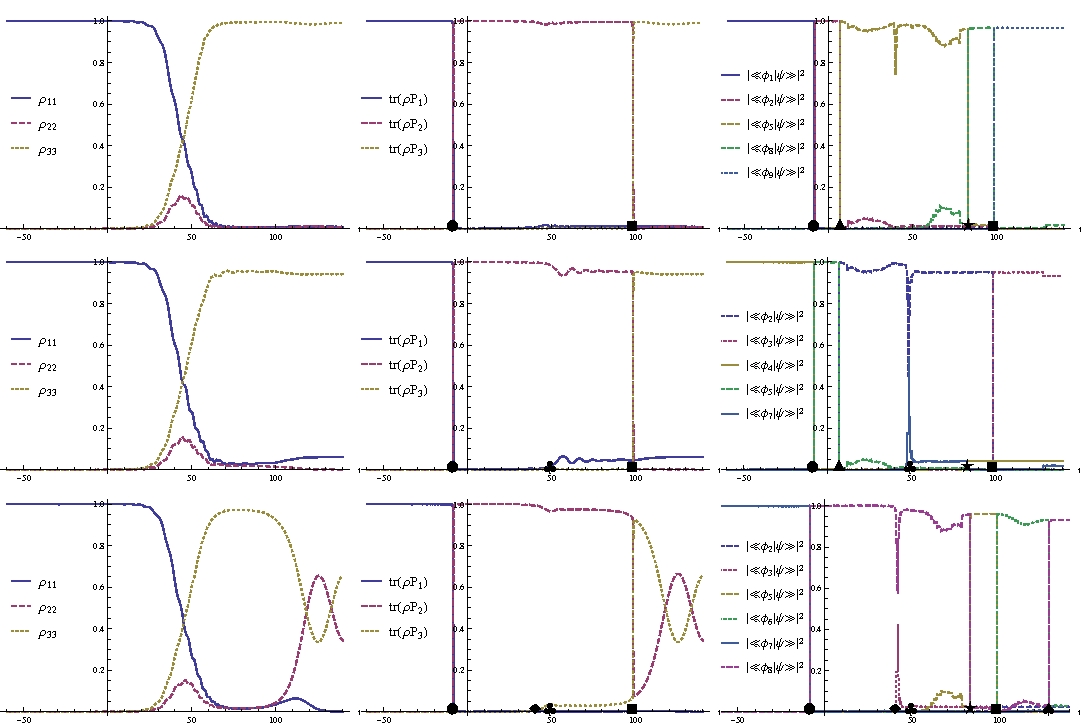}
\caption{\label{occ_prob_static1} Left: Occupation probabilities of the bare states $\rho_{\psi,ii}$ with respect to $t$. Middle: occupation probabilities of the instantenous eigenstates of the system $\tr_\S(\rho_\psi P_i)$ with respect to $t$. Right: occupation probabilities of the instantenous eigenstates of the universe $|\llangle \phi_a|\psi \rrangle|^2$ with respect to $t$. Up: for the initial condition $\psi(t_0) = \phi_1(0)$. Middle: for the initial condition $\psi(t_0) = \phi_4(0)$. Down: for the initial condition $\psi(t_0) = \phi_7(0)$. The event $\bullet$ is the transition $\zeta_1 \to \zeta_2$ and the event $\blacksquare$ is the transition $\zeta_2 \to \zeta_3$. The event $\blacktriangle$ is the transition $\phi_2 \leftrightarrow \phi_5$, the event $\star$ is the transition $\phi_5 \leftrightarrow \phi_8$, and the event $\spadesuit$ is the transition $\phi_6 \to \phi_8$. The events $\blacklozenge$ and $\clubsuit$ are the beginnings of non-adiabatic exchanges in the universe.}
\end{center}
\end{figure}

\begin{figure}
\begin{center}
\includegraphics[width=12.8cm]{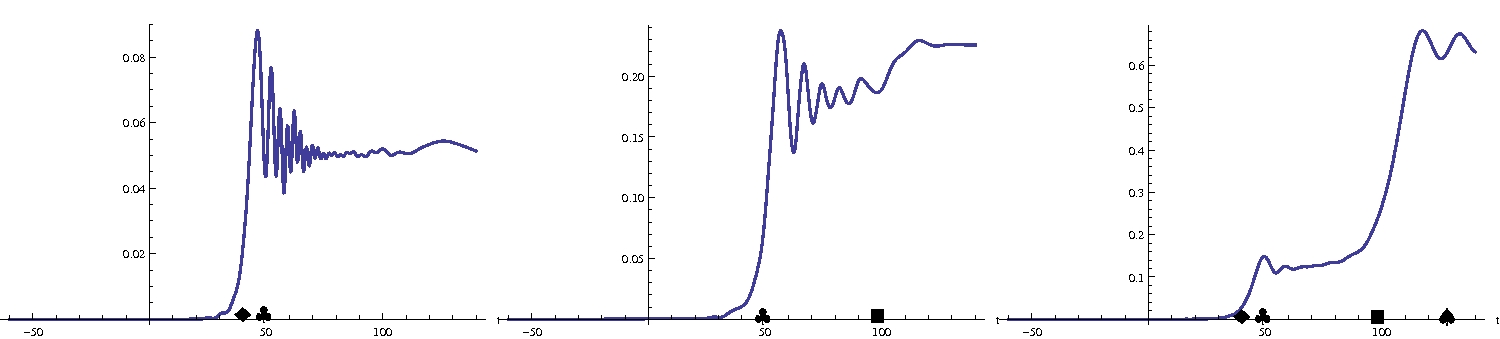}
\caption{\label{entropy_static} von Neuman entropy of the density matrix $- \tr_\S (\rho_\psi \ln \rho_\psi)$ with respect to $t$ for the initial conditions $\psi(t_0) = \phi_1(0)$, $\psi(t_0) = \phi_4(0)$ and $\psi(t_0) = \phi_7(0)$.}
\end{center}
\end{figure}
In the first case the control is unperfectly realized, in the second case the control quality is very small, and in the last case the control completely fails. We want to interpret geometrically the different dynamical effects appearing in the previous figures. We first note that the average curving is zero for all states. Figure \ref{fake_static} shows the average fake curvatures and figure \ref{geoentropy_static} shows the eigenentropies.
\begin{figure}
\begin{center}
\includegraphics[width=15.5cm]{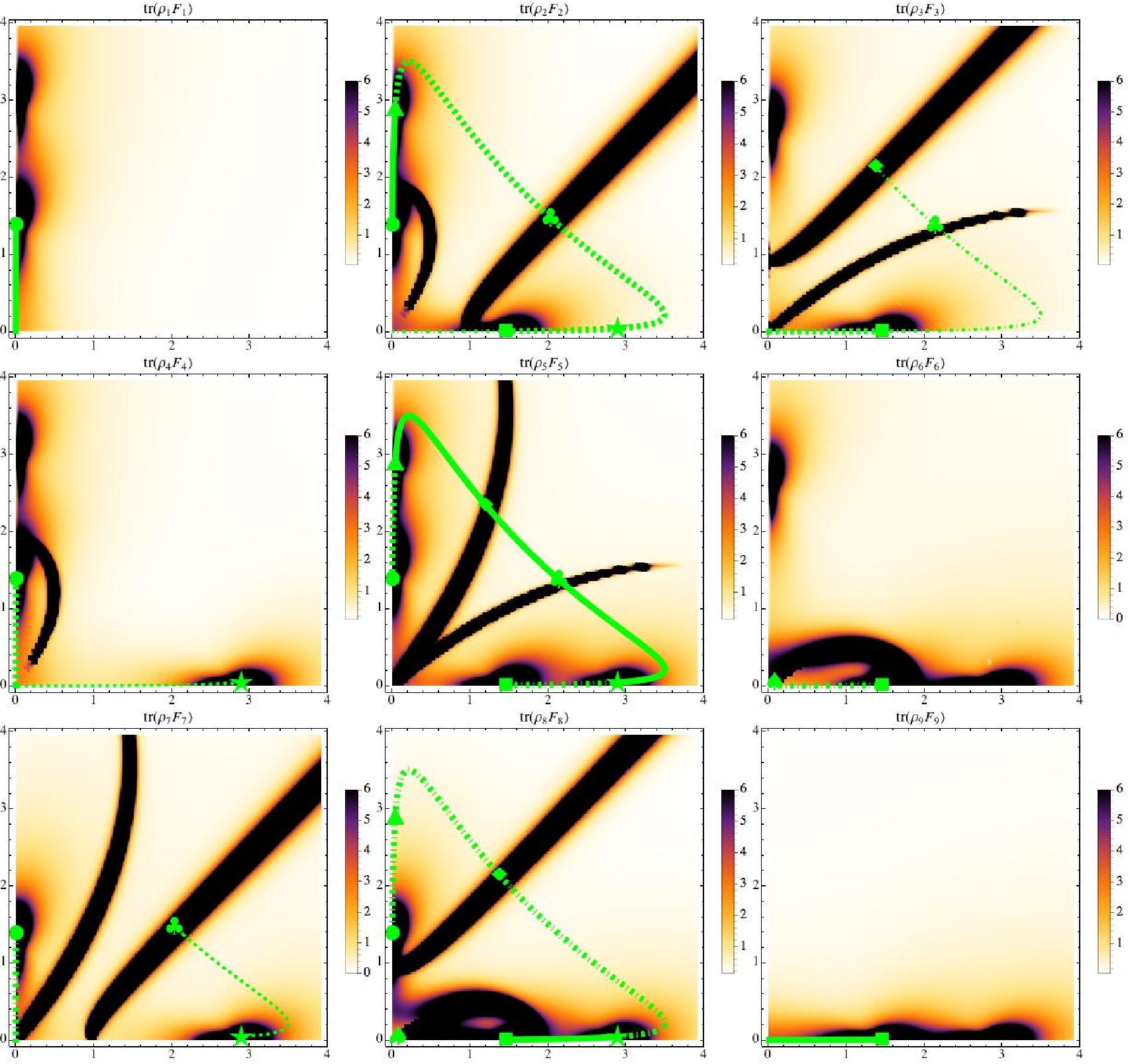}
\caption{\label{fake_static} Densities of the average fake curvature $\tr_\S(\rho_a F_a)$ in $M$, with the path $\mathcal C$ (plain line for the case 1, dashed line for the case 2, dotted line for the case 3; large line for the states with a large occupation and thin line for the states with a small occupation). We show the events $\bullet$ (transition $\zeta_1 \to \zeta_2$), $\blacktriangle$ (transition $\phi_2 \leftrightarrow \phi_5$), $\blacklozenge$ and $\clubsuit$ (beginnings of non-adiabatic exchanges), $\star$ (transition $\phi_5 \leftrightarrow \phi_8$), $\blacksquare$ (transition $\zeta_2 \to \zeta_3$), and $\spadesuit$ (transition $\phi_6 \to \phi_8$).}
\end{center}
\end{figure}
\begin{figure}
\begin{center}
\includegraphics[width=15.5cm]{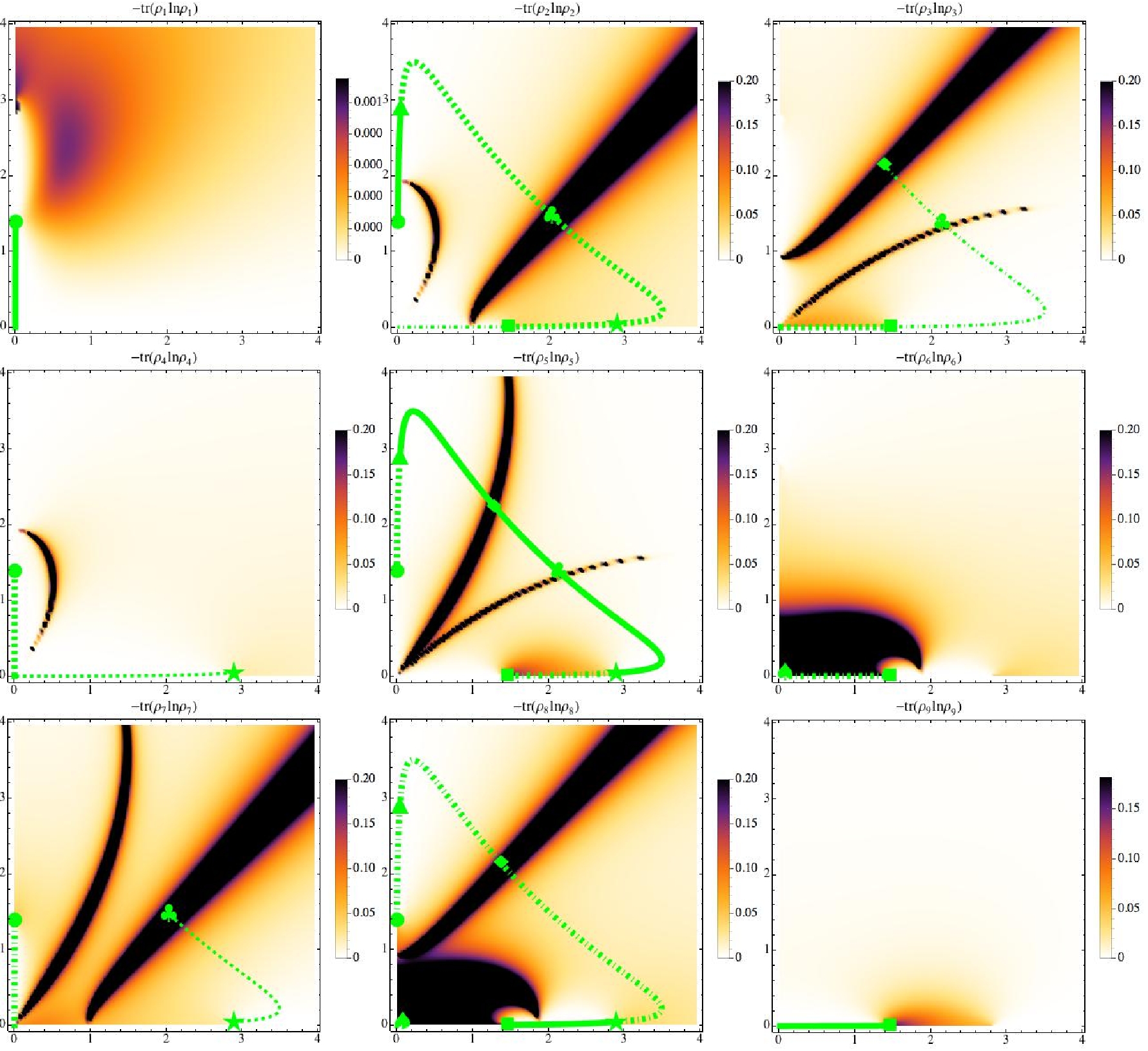}
\caption{\label{geoentropy_static} Eigenentropies $-\tr_\S(\rho_a \ln \rho_a)$ in $M$, with the path $\mathcal C$ (plain line for the case 1, dashed line for the case 2, dotted line for the case 3; large line for the states with a large occupation and thin line for the states with a small occupation). We show the events $\bullet$ (transition $\zeta_1 \to \zeta_2$), $\blacktriangle$ (transition $\phi_2 \leftrightarrow \phi_5$), $\blacklozenge$ and $\clubsuit$ (beginnings of the increase of the entropy of the system), $\star$ (transition $\phi_5 \leftrightarrow \phi_8$), $\blacksquare$ (transition $\zeta_2 \to \zeta_3$), and $\spadesuit$ (transition $\phi_6 \to \phi_8$).}
\end{center}
\end{figure}
We see on these figures that the regions with strong eigenentropies are correlated with some regions with strong fake curvature (which have not the morphology of a point singularity). Because of their morphology, we call such a region an \textit{entropic string}. We summarize the different events:
\begin{itemize}
\item[$\bullet$] the path $\mathcal C$ passes through a singularity of $\tr_\S (\rho_a F_a)$ and $\tr_S (\rho_{a+1} F_{a+1})$ (with $a=1$ for the firts case, $a=4$ for the second one, and $a=7$ for the last case) which induces a transition $P_1 \to P_2$ for the system;
\item[$\blacktriangle$] (not for the last case) $\mathcal C$ passes through a singularity of $\tr_\S (\rho_2 F_2)$ and $\tr_S (\rho_5 F_5)$ which induces a transition $\phi_2 \leftrightarrow \phi_5$ (without non-adiabatic effects for the system);
\item[$\blacklozenge$] (not for the second case) $\mathcal C$ passes through an entropic string (of $\tr_\S (\rho_5 F_5)$ for the first case, and of $\tr_\S(\rho_3 F_3)$ and $\tr_S(\rho_8 F_8)$ for the last case) which induces the increase of the entropy of the system (and non-adiabatic exchanges between $\phi_3$ and $\phi_8$ for the last case);
\item[$\clubsuit$] $\mathcal C$ passes through an entropic string (of $\tr_\S (\rho_5 F_5)$ for the first case, of $\tr_\S (\rho_2 F_2)$ and $\tr_\S (\rho_7 F_7)$ for the second case, and of $\tr_\S (\rho_3 F_3)$ and $\tr_\S(\rho_5 F_5)$ for the last case) which induces oscillations of the entropy of the system (with non-adiabatic exchanges);
\item[$\star$] (not for the second case) $\mathcal C$ passes through a singularity of $\tr_\S (\rho_5 F_5)$ and $\tr_\S (\rho_8 F_8)$ which induces a transition $\phi_5 \leftrightarrow \phi_8$;
\item[$\blacksquare$] $\mathcal C$ passes through a singularity of $\tr_\S (\rho_8 F_8)$ and $\tr_\S(\rho_9 F_9)$ (first case), or of $\tr_\S (\rho_2 F_2)$ and $\tr_\S(\rho_3 F_3)$ (second case), or of $\tr_\S (\rho_5 F_5)$ and $\tr_\S(\rho_6 F_6)$ (last case), which induces a transition $P_2 \to P_3$ for the system; for the last case only $\mathcal C$ enters in a strong entropic string of $\tr_\S (\rho_6 F_6)$ and $\tr_\S(\rho_8 F_8)$ which induces strong non-adiabatic exchanges.
\item[$\spadesuit$] (only for the last case) $\mathcal C$ passes through a singularity of $\tr_\S (\rho_6 F_6)$ and $\tr_\S(\rho_8 F_8)$ which induces a transition $\phi_6 \to \phi_8$.
\end{itemize}
The passages by the entropic strings ($\blacklozenge$, $\clubsuit$ and $\blacksquare$ in the last case) are responsible of the failures of the adiabatic quantum control. We see two kinds of singularities of the fake curvature, the first one ($\bullet$ and $\blacksquare$) induces adiabatic passages for the system, the second one ($\blacktriangle$, $\star$ and $\spadesuit$) induces adiabatic passages for the environment without effects on the system. The inactive and active singularities (from the viewpoint of the controlled system) seem not to be easily distinguishable without comparison with the singularities of the isolated system. The average curving being equal to zero in this case, it does not eliminate the inactive singularities as in the previous case. This is clearly a drawback of an analysis based on adiabatic fields.\\
It is possible to give an heuristic explanation of these results. Write $\rho_a(x) = \sum_{i,j} p_{ij}(x) |\zeta_i(x) \rangle \langle \zeta_j(x)|$. The exact calculations of $p_{ij}(x)$ and of the different fields can be complicated. But by analogy with section \ref{schmidt}, we can heuristically suppose that $\tr_\S (\rho_a F_a)$ and $\tr_\S(\rho_a B_a)$ have behaviours similar to $\tr_{\mathbb C^3} (\varpi^2(F_{\mathcal S} + [A_{\mathcal S},\hat A_{\mathcal E}^t]))$ and $\tr_{\mathbb C^3} (\varpi^2(F_{\mathcal E} - A_{\mathcal S} \wedge A_{\mathcal S}))$ with $(\varpi^2)_{ij} = p_{ij}$. $A_{\mathcal S} \in \Omega^1(M,\mathfrak u(3))$ is the non-abelian generator of the geometric phase for all states of the system, and we have then $F_{\mathcal S} = dA_{\mathcal S} + A_{\mathcal S} \wedge A_{\mathcal S} = 0$. By the same manner $F_{\mathcal E} = 0$. The adiabatic transitions inner to $\S$ do not induce decoherence processes, we can then expect that $\tr_{\mathbb C^3} (\varpi^2 A_{\mathcal S} \wedge A_{\mathcal S})$ and then $\tr_\S(\rho_a B_a)$ are zero. Moreover $\tr_{\mathbb C^3} (\varpi^2 [A_{\mathcal S},\hat A_{\mathcal E}^t])$ feels essentially the non-adiabatic couplings for the system ($A_{\mathcal S}$) and for the environment ($A_{\mathcal E}$). Since the coupling is perturbative we can expect that $\tr_\S (\rho_a F_a)$ presents essentially the singularities associated with the state $\lim_{g \to 0} \phi_a$ which characterizes non-adiabatic couplings for the system (active singularities) and for the environment (passive singularities). This argument is not a proof, it is just a heuristic argument showing the consistency of the numerical results with the interpretations of the adiabatic fields.

\subsection{Two atoms with a dynamical coupling}
We consider now a dynamical coupling ($g=0.1\ ua$): $V(x) = g(|\zeta_2(x)\otimes \xi_3(x)\rrangle \llangle \zeta_3(x)\otimes \xi_2(x)| + |\zeta_3(x)\otimes \xi_2(x)\rrangle \llangle \zeta_2(x)\otimes \xi_3(x)|)$, i.e. the potential couples the dressed states in place of the bare states. In this case the density eigenmatrices have a rank equal to 1 except $\rho_6$ and $\rho_8$ which have a rank equal to 2 (for non zero laser intensities). We consider two initial conditions $\psi(t_0)= \phi_1(0)$ and $\psi(t_0) = \phi_7(0)$. The occupation probabilities are drawn figure \ref{occ_prob_dynamic} and the entropies of the density matrix are drawn figure \ref{entropy_dynamic}.
\begin{figure}
\begin{center}
\includegraphics[width=12.8cm]{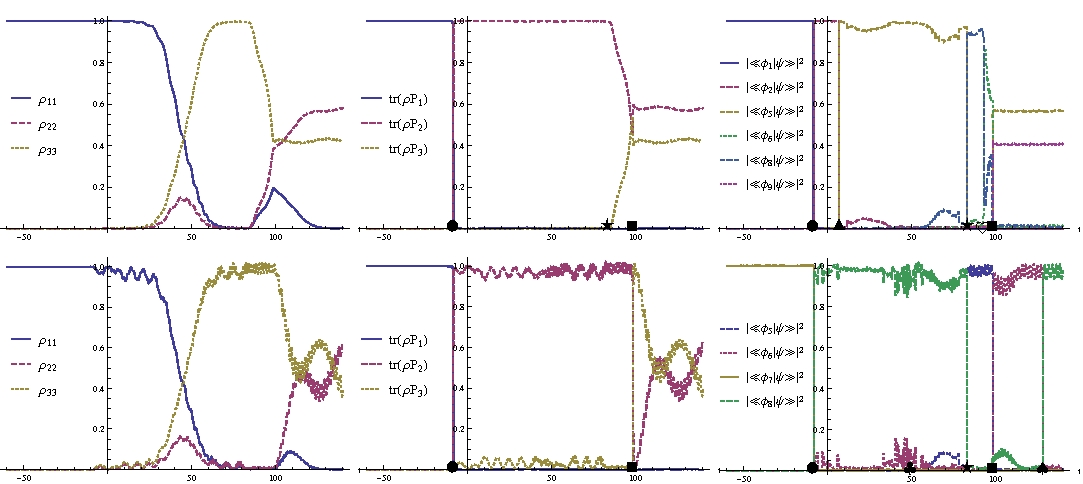}
\caption{\label{occ_prob_dynamic} Left: occupation probabilities of the bare states $\rho_{\psi,ii}$ with respect to $t$. Middle: occupation probabilities of the instantenous eigenstates of the system $\tr_\S(\rho_\psi P_i)$ with respect to $t$. Right: occupation probabilities of the instantenous eigenstates of the universe $|\llangle \phi_a|\psi \rrangle|^2$ with respect to $t$. Up: for the initial condition $\psi(t_0) = \phi_1(0)$. Down: for the initial condition $\psi(t_0) = \phi_7(0)$. The event $\bullet$ is the transition $\zeta_1 \to \zeta_2$ and the event $\blacksquare$ is the transition $\zeta_2 \to \zeta_3$. The event $\blacktriangle$ is the transition $\phi_2 \leftrightarrow \phi_5$, the event $\star$ is the transition $\phi_5 \leftrightarrow \phi_8$, and the events $\spadesuit$ and $\heartsuit$ are the transitions $\phi_6 \leftrightarrow \phi_8$. The events $\blacklozenge$ and $\clubsuit$ are the beginnings of non-adiabatic exchanges in the universe.}
\end{center}
\end{figure}
\begin{figure}
\begin{center}
\includegraphics[width=12.8cm]{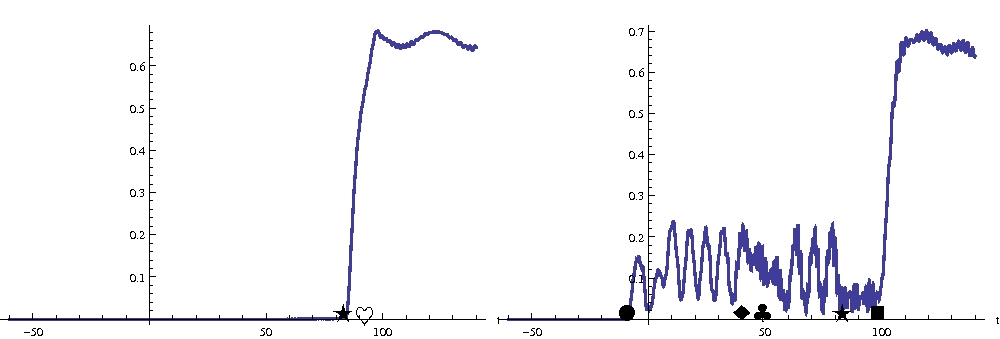}
\caption{\label{entropy_dynamic} von Neuman entropy of the density matrix $- \tr_\S (\rho_\psi \ln \rho_\psi)$ with respect to $t$ for the initial conditions $\psi(t_0) = \phi_1(0)$ (left) and $\psi(t_0) = \phi_7(0)$ (right).}
\end{center}
\end{figure}
In these two cases, the control dramatically fails. The average fake curvatures are shown figure \ref{fake_dynamic}, the average curving are shown figure \ref{curving_dynamic} and the eigenentropies are shown figure \ref{geoentropy_dynamic}.
\begin{figure}
\begin{center}
\includegraphics[width=15.5cm]{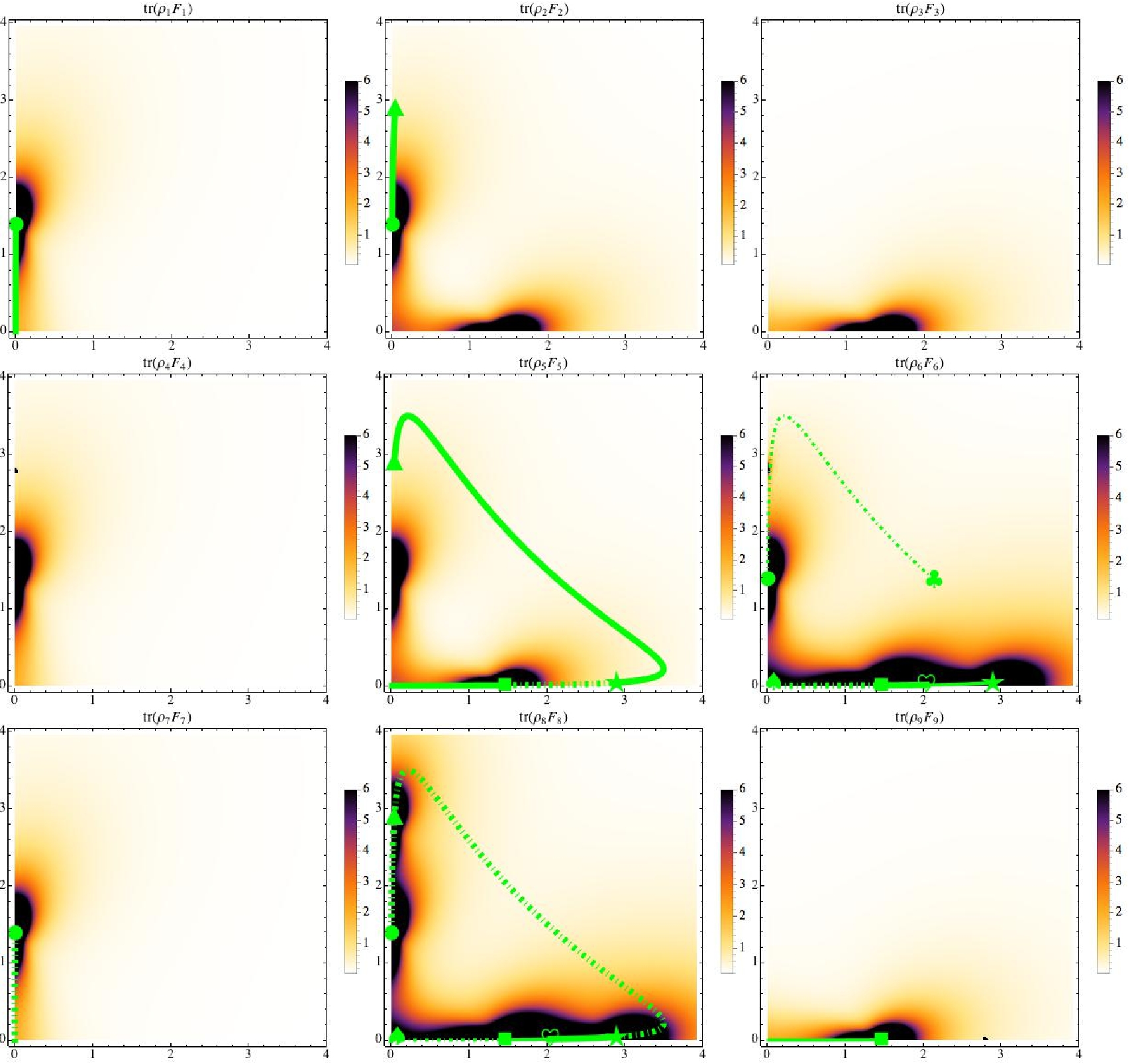}
\caption{\label{fake_dynamic} Densities of the average fake curvature $\tr_\S(\rho_a F_a)$ in $M$, with the path $\mathcal C$ (plain line for the case 1 and dotted line for the case 2; large line for the states with a large occupation and thin line for the states with a small occupation). We show the events $\bullet$ (transition $\zeta_1 \to \zeta_2$), $\blacktriangle$ (transition $\phi_2 \leftrightarrow \phi_5$), $\blacklozenge$ and $\clubsuit$ (beginnings of non-adiabatic exchanges), $\star$ (transition $\phi_5 \leftrightarrow \phi_8$), $\blacksquare$ (transition $\zeta_2 \to \zeta_3$), and $\spadesuit$ and $\heartsuit$ (transition $\phi_6 \leftrightarrow \phi_8$).}
\end{center}
\end{figure}
\begin{figure}
\begin{center}
\includegraphics[width=15.5cm]{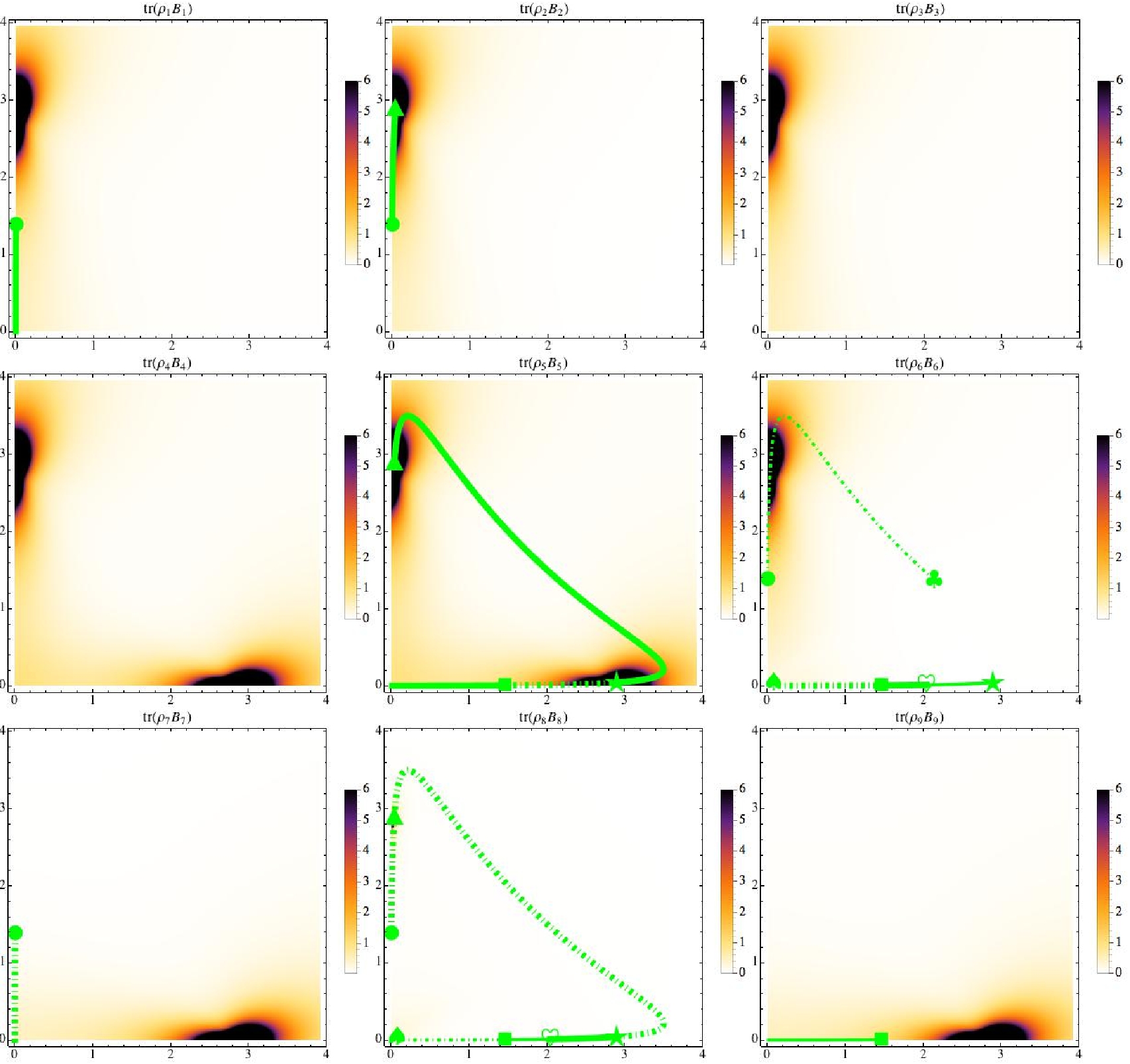}
\caption{\label{curving_dynamic} Densities of the average curving $\tr_\S(\rho_a B_a)$ in $M$,with the path $\mathcal C$ (plain line for the case 1 and dotted line for the case 2; large line for the states with a large occupation and thin line for the states with a small occupation). We show the events $\bullet$ (transition $\zeta_1 \to \zeta_2$), $\blacktriangle$ (transition $\phi_2 \leftrightarrow \phi_5$), $\blacklozenge$ and $\clubsuit$ (beginnings of non-adiabatic exchanges), $\star$ (transition $\phi_5 \leftrightarrow \phi_8$), $\blacksquare$ (transition $\zeta_2 \to \zeta_3$), $\spadesuit$ and $\heartsuit$ (transition $\phi_6 \leftrightarrow \phi_8$). }
\end{center}
\end{figure}
\begin{figure}
\begin{center}
\includegraphics[width=5.2cm]{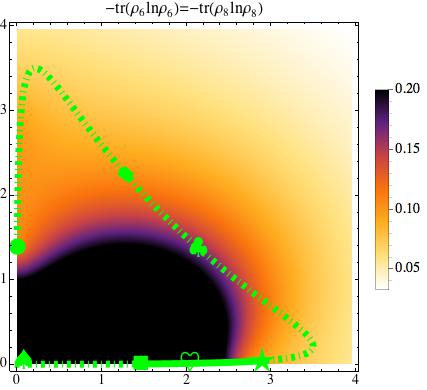}
\caption{\label{geoentropy_dynamic} Eigenentropies $-\tr_\S(\rho_6 \ln \rho_6) = -\tr_\S(\rho_8 \ln \rho_8)$ in $M$ (the other eigenentropies are zero), with the path $\mathcal C$ (plain line for the case 1 and dotted line for the case 2) and the events $\bullet$ (transition $\zeta_1 \to \zeta_2$), $\blacktriangle$ (transition $\phi_2 \leftrightarrow \phi_5$), $\blacklozenge$ and $\clubsuit$ (beginnings of the increase of the entropy of the system), $\star$ (transition $\phi_5 \leftrightarrow \phi_8$) and $\blacksquare$ (transition $\zeta_2 \to \zeta_3$).}
\end{center}
\end{figure}
We summarize the different events:
\begin{itemize}
\item[$\bullet$] the path $\mathcal C$ passes through a singularity of $\tr_\S (\rho_1 F_1)$ and $\tr_S (\rho_{2} F_{2})$ (or of $\tr_\S (\rho_7 F_7)$ and $\tr_S (\rho_{8} F_{8})$ for the second case) which induces a transition $P_1 \to P_2$ for the system;
\item[$\blacktriangle$] (only for the first case) $\mathcal C$ passes through a singularity of $\tr_\S (\rho_2 B_2)$ and $\tr_S (\rho_5 B_5)$ which induces a transition $\phi_2 \leftrightarrow \phi_5$ (without non-adiabatic effects for the system);
\item[$\blacklozenge$] (only for the second case) $\mathcal C$ approaches the region with large eigenentropies $-\tr(\rho_6 \ln \rho_6)$ and $-\tr(\rho_8 \ln \rho_8)$, which induces strong non-adiabatic oscillations between $\phi_6$ and $\phi_8$;
\item[$\clubsuit$] (only for the second case) $\mathcal C$ passes close to the region with large eigenentropies $-\tr(\rho_6 \ln \rho_6)$ and $-\tr(\rho_8 \ln \rho_8)$, the non-adiabatic oscillations present a maximal amplitude;
\item[$\star$]  $\mathcal C$ passes through a singularity of $\tr_\S (\rho_5 B_5)$ and $\tr_\S (\rho_8 F_8)$ which induces a transition $\phi_5 \leftrightarrow \phi_8$; with the increase of the entropy of the system (with non-adiabatic exchanges) for the first case and with a modification of the entropy oscillations in the second case;
\item[$\heartsuit$] (only for the first case) $\mathcal C$ passes through a singularity of $\tr_\S (\rho_6 F_6)$ and $\tr_\S(\rho_8 F_8)$ which induces a transition $\phi_8 \to \phi_6$.
\item[$\blacksquare$] $\mathcal C$ passes through a singularity of $\tr_\S (\rho_5 F_5)$, $\tr_\S (\rho_6 F_6)$, $\tr_\S (\rho_8 F_8)$ and $\tr_\S(\rho_9 F_9)$ (first case), or of $\tr_\S (\rho_7 F_7)$ and $\tr_\S(\rho_8 F_8)$ (second case), which induces a transition $P_2 \to P_3$ for the system; for the second case $\mathcal C$ passes through the region with large eigenentropies $-\tr(\rho_6 \ln \rho_6)$ and $-\tr(\rho_8 \ln \rho_8)$ which induces a strong increase of the entropy of the system with strong non-adiabatic exchanges.
\item[$\spadesuit$] (only for the second case) $\mathcal C$ passes through a singularity of $\tr_\S (\rho_6 F_6)$ and $\tr_\S(\rho_8 F_8)$ which induces a transition $\phi_6 \to \phi_8$.
\end{itemize}

\section{Discussion and conclusion}
\subsection{Larger environments}
The results presented section 3 are independent of the dimensions of the Hilbert spaces of the system and of the environment. But the example treated section 3 concerns a small system (three level atom) with a very small environment (another three level atom). The restriction to few levels for the system is not drastic. By principle, adiabatic quantum control consists to control the occupation probabilities of few levels. For $N$ level system (with possibly $N$ very large) we can consider only few parameter dependent eigenlevels linked by crossings or avoid crossings to the initial one. The adiabatic elimination of the other states is valid while the control path does not approach the crossings between a selected eigenlevel and an eliminated eigenlevel. By the property 2, we know that the (fake) curvature presents a singularity at such a point for the implicated selected eigenlevel and which is not correlated with a singularity of another selected eigenlevel (since the crossing implicated an eliminated eigenlevel). It is then easy to locate on the ``density charts'' these forbidden regions (associated with fails of the adiabatic assumption).\\

For a complicated environment (with a lot of quantum levels) the analysis exposed in the present paper could be difficult to realize without other assumptions because the number of ``density charts'' of adiabatic field strengths to study becomes very large. Moreover, in a lot of situations, a complete and exact description of the environment (its Hamiltonian and the coupling operator) is unknown. To solve the first problem, we can consider the possibility of using effective small environments reproducing the effects of a large environment on the quantum system. For the second one, we can work only at the stage of the density matrices. We briefly discuss about these two possibilities in this section. 

\subsubsection{Effective Hamiltonians for large environment: }
If $\dim \E$ is very large, the number of ``density charts'' to draw becomes too large to be realized. To solve this problem, we can proceed, in a first approximation, by adiabatic eliminations. Let $\{\xi_\alpha(x)\}_{\alpha \in \{1,..., \dim \E\}}$ be the parameter dependent eigenvectors of $H_{\mathcal E}(x)$. We assume that for physical reasons, only a few of these states $\{\xi_\alpha(x)\}_{\alpha \in I}$ are strongly implicated at the starting point of the control $x_0 \in M$ (each considered control path starting and ending at $x_0$). We can then reduce the number of the environment states by projections onto the space spanned by $\{\xi_\alpha(x)\}_{\alpha \in I}$. In other words, we consider (for the eigenvector problem) the effective Hamiltonian $H^{eff,0}_{\mathcal U}(x) = H_{\mathcal S}(x) \otimes 1_\E + 1_{\S} \otimes P_I(x) H_{\mathcal E}(x) P_I(x) + 1_\S \otimes P_I(x) V(x) 1_\S \otimes P_I(x)$ where $P_I(x)$ is the orthogonal projection onto the space spanned by $\{\xi_\alpha(x)\}_{\alpha \in I}$ (for the dynamics, it is associated with the effective Hamiltonian $H^{eff-dyn,0}_{\mathcal U} = H^{eff,0}_{\mathcal U} - \ihbar 1_\S \otimes \dot P_I P_I$ -- the effective theories involve two effective Hamiltonians, one for the computation of the effective eigenvectors and the other for the computation of the effective dynamics --). If the control is adiabatic with respect to the environment dynamics, this approximation is valid while the control path does not approach a crossing between a selected and an eliminated environment levels. The forbidden regions of $M$ by this requirement are characterized by singularities of the curving (by the relation between the curving and the environment curvature and the property 2). We can then locate these forbidden regions on the ``density charts''.\\

For example if we suppose that $x_0$ corresponds to the situation where the system and the environment are free (the control apparatus does not act on the quantum objects, then we can suppose that a very large environment is in a thermal equilibrium state : $ \rho_{\mathcal E} = \frac{e^{- \beta H_{\mathcal E}(x_0)}}{Z}$ (where $\beta = \frac{1}{k_B T}$, $T$ being the temperature and $k_B$ is the Boltzmann constant), with $Z = \tr_\E e^{- \beta H_{\mathcal E}(x_0)}$. The different initial conditions are then
\begin{equation}
\rho_{\mathcal U,i}(0) = \sum_{\alpha =1}^{\dim \E} \frac{e^{-\beta \nu_\alpha(x_0)}}{Z} |\phi_{a(i,\alpha)}(x_0) \rrangle \llangle \phi_{a(i,\alpha)}(x_0)|
\end{equation}
where $\nu_\alpha$ is the eigenvalue associated with $\xi_\alpha$ and $\phi_{a(i,\alpha)}$ is the eigenstate of $H_{\mathcal U}$ such that
\begin{equation}
\lim_{g \to 0} \phi_{a(i,\alpha)} = \zeta_i \otimes \xi_\alpha
\end{equation}
with $\{\zeta_i(x)\}$ the eigenstates of $H_{\mathcal S}$ and $g$ is the weak system-environment coupling amplitude. $\tr_\E \rho_{\mathcal U,i}(0) = \rho_{\mathcal E} + \mathcal O(g)$. The density matrix of the system with respect to the time is then
\begin{eqnarray}
& & \rho_{\mathcal S,i}(t) \nonumber \\
& & = \sum_{\alpha=1}^{\dim \E}  \frac{e^{-\beta \nu_\alpha(x_0)}}{Z} \tr_\E \left(U(t,0) |\phi_{a(i,\alpha)}(x_0) \rrangle \llangle \phi_{a(i,\alpha)}(x_0)| U(t,0)^\dagger \right)
\end{eqnarray}
where $U(t,0)$ is the evolution operator of the universe, i.e. the solution of $\ihbar \frac{dU(t,0)}{dt} = H_{\mathcal U}(x(t)) U(t,0)$ with $U(0,0) = 1_{\S \otimes \E}$. We can consider only the states $\{\xi_\alpha\}$ associated with the smallest energies (the Boltzmann factors $e^{-\beta \nu_\alpha(x_0)}$ of the others being negligible) and with the assumption that the control path does not approach crossings between the selected and the eliminated states, we have
\begin{eqnarray}
& & \rho_{\mathcal S,i}(t) \simeq  \sum_{\alpha \in I}  \frac{e^{-\beta \nu_\alpha(x_0)}}{Z} \\
& & \quad \times \tr_\E \left(U^{eff,0}(t,0) |\phi_{a(i,\alpha)}(x_0) \rrangle \llangle \phi_{a(i,\alpha)}(x_0)| U^{eff,0}(t,0)^\dagger \right)
\end{eqnarray}
with $\ihbar \frac{dU^{eff,0}(t,0)}{dt} = H_{\mathcal U}^{eff,0}(x(t)) U^{eff,0}(t,0)$. The control of $\rho_{\mathcal S,i}(t)$ can then be deduced from the control of the states $\{U(t,0) \phi_{a(i,\alpha)}(x_0)\}_{i\in\{1,...,\dim \S\},\alpha \in I}$ at the adiabatic limit with reasonable number of density charts to consider.\\

Nevertheless, the replacement of the universe Hamiltonian by $H^{eff,0}_{\mathcal U}$ is certainly too drastic for a lot of situations. We can indeed think that the requirement for the control to be strictly adiabatic with respect to the environment dynamics is difficult to assume since the ``experimentalist'' which controls the system ``loses'' the informations concerning the environment (a lost of information whih is mathematically modeled by the partial trace on $\E$). Recently, we have proposed a framework to treat almost adiabatic dynamics \cite{Viennot5}. It consists to consider (for the eigenvector problem) the effective Hamiltonian $H^{eff,1}_{\mathcal U} = 1_\S \otimes P_I H_{\mathcal U} \Omega$ where $\Omega$ (called the wave operator) is solution of the Bloch equation $[H_{\mathcal U},\Omega]\Omega = 0$ with $\Omega 1_\S \otimes P_I = \Omega$ (a complete exposition of the wave operator theory can be found in \cite{Killingbeck,Jolicard}, the effective Hamiltonian for the dynamics being $H^{eff-dyn,1}_{\mathcal U} = H^{eff,1}_{\mathcal U} - \ihbar 1_S \otimes \dot P_I \Omega$). The interest of $H^{eff,1}_{\mathcal U}$ with respect to $H^{eff,0}_{\mathcal U}$ is that $H^{eff,1}_{\mathcal U}$ takes into account non-adiabatic dynamical effects induced by the outer of the space spanned by $\{\xi_\alpha(x)\}_{\alpha \in I}$. It takes into account corrections of the strict adiabatic approximation. It is important to note that $H^{eff,1}_{\mathcal U}$ is not self-adjoint and generates then a non-unitary evolution. This is not an artefact and it is one of the corrections of the strict adiabaticity. A decrease of the norm is associated with transitions from $\S \otimes \Ran P_I$ to $\S \otimes \ker P_I$ and an increase is associated with converse transitions. This effective Hamiltonian cannot include precise information about the eliminated states but it includes informations concerning ``quantum flows'' entering and exiting from the $\S \otimes \Ran P_I$. We can compute the fake curvature and the curving by using the eigenstates of $H^{eff,1}_{\mathcal U}$ to draw a small number of ``density charts'', which permit to analyse the control problem (because of the non-selfadjointness of $H^{eff,1}_{\mathcal U}$ it needs to redefine the generator of the $C^*$-geometric phase by $\mathcal A_a = \tr_\E (|d\phi_a \rrangle \llangle \phi_a^*|) \rho_a^{-1}$ with $\rho_a = \tr_\E|\phi_a \rrangle \llangle \phi_a^*|$ where $\{\phi_a^*\}_a$ are the eigenvectors of ${H^{eff,1}_{\mathcal U}}^\dagger$ to take into account the biorthogonality of the eigenbasis).\\

We note that the introduction of an effective Hamiltonian to reduce the size of the Hilbert space of the universe (needed to have a small number of density charts to consider) introduces few approximations in the description of the system. But the consideration of the density charts of the curving and of the fake curvature is not used to obtain quantitative results but only qualitative results. As in the example section 4, we can use it \textit{a posteriori} to analyse and to interpret the hampering on the control by the entanglement with the environment, and to understand the different processes occurring. And we can use it \textit{a priori} to find the rough shape of a control solution path taking into account the hampering by the environment (such a solution could be next optimized by specific numerical methods \cite{Bonnard}).

\subsubsection{Working directly at the stage of the density matrices:} The density matrix of the system is solution of the equation
\begin{equation}
\ihbar \dot \rho_{\mathcal S}(t) = \mathcal L_{x(t)}(\rho_{\mathcal S}(t))
\end{equation}
where in the context used in this paper $\mathcal L_{x(t)}(\rho_{\mathcal S}(t)) = \tr_\E [H_{\mathcal U}(x(t)) , |\psi(t) \rrangle \llangle \psi(t)| ]$. But in a lot of situations, a complete description of $H_{\mathcal U}$ is unknown or it is not practical to use. Under some assumptions (in particular with a Markovian approximation) the map $\mathcal L_x$ takes the Lindblad form \cite{Breuer} :
\begin{eqnarray}
\mathcal L_x(\rho) & = & [\tilde H_{\mathcal S}(x),\rho] + \imath \sum_k \gamma_k(x) \Gamma_k(x) \rho \Gamma_k(x)^\dagger \nonumber \\
& & \quad - \frac{\imath}{2} \sum_k \gamma_k(x) \{\Gamma_k(x)^\dagger \Gamma_k(x), \rho\}
\end{eqnarray}
$\gamma_k \in \mathbb R$, $\Gamma_k \in \mathcal L(\S)$, $\tilde H_{\mathcal S}$ being a system Hamiltonian possibly different from the free Hamiltonian $H_{\mathcal S}$ ($\{.,.\}$ being the anticommutator). It is interesting to note that the density eigenmatrix $\rho_a$ when it is induced by an eigenvector of $H_{\mathcal U}$ satisfies
\begin{equation}
\label{eqMD1}
\mathcal L_x(\rho_a(x)) = 0
\end{equation}
$\rho_a$ is a steady state of the system. We postulate that even if $H_{\mathcal U}$ is unknown or forgotten, that the density eigenmatrix are still parameter dependent steady state of the system governed by $\mathcal L_x$. Starting with a steady state $\rho_a(x(0))$, the dynamics, if it is sufficiently adiabatic, follows $\rho_a(x(t))$ except in the regions of rapid adiabatic passages (corresponding to the crossings of eigenlevels of $H_{\mathcal U}$) and in the regions inducing local or kinematic decoherence. Since we cannot use $H_{\mathcal U}$ which is unknown or forgotten, it needs to compute the curving and the fake curvature directly from $\rho_a$. Singularities and entropic strings of these fields indicate these regions of $M$. The generator of the $C^*$-geometric phase is solution of the equation \cite{Viennot1}
\begin{equation}
\label{eqMD2}
d \rho_a = \mathcal A_a \rho_a + \rho_a \mathcal A_a^\dagger
\end{equation}
After solving this equation we have $B_a = d\mathcal A_a - \mathcal A_a \wedge \mathcal A_a$. For eigenvectors $\phi_a$ of $H_{\mathcal U}$ associated with a non-degenerate eigenvalues (as in the example section 4), we have by construction $A_a = \llangle \phi_a|d\phi_a \rrangle 1_{\S}$. But $\llangle \phi_a|d\phi_a \rrangle = \tr_\S(\rho_a \mathcal A_a)$ (see \cite{Viennot2}), we can then compute directly the reduced potential as
\begin{equation}
\label{eqMD3}
A_a = \tr_\S(\rho_a \mathcal A_a) 1_\S
\end{equation}
and then $F_a = d \tr_\S (\rho_a \mathcal A_a) 1_\S - B_a$.\\
By solving equations \ref{eqMD1}, \ref{eqMD2} and \ref{eqMD3} we can finally draw ``density charts'' of the fields of the higher gauge structure by working only at the stage of the density matrix formalism.\\
Other approaches of adiabatic quantum control of open systems \cite{Sarandy1, Sarandy2, Pekola} are based on the consideration of a density matrix as being an Hilbert-Schmidt operator and $\mathcal L_x$ as being a non-selfadjoint operator on the Hilbert-Schmidt space (the Hilbert-Schmidt space is so-called Liouville space and the operators on the Hilbert-Schmidt space are so-called ``superoperators''). We are interested by a comparison of these adiabatic approaches with the present one because they are also completely generic: as our approach, they do not require a particular form of the control Hamiltonian and they can be used with several independent control parameters (adiabatic approaches with a single control parameter have a trivial differential geometry and cannot be analysed by a geometric framework since the problem is ``under-parametrized'', in other words the Hamiltonian matrix of the problem is not ``versal'' in the sense of the Arnold's theory \cite{Arnold}). It seems then pertinent to compare the present approach with the adiabatic quantum control methods based on the Hilbert-Schmidt representation because they have the same possibilities of geometric analyses. The density eigenmatrices considered in the works \cite{Sarandy1, Sarandy2, Pekola} are eigenvectors of $\mathcal L_x$ in the Hilbert-Schmidt sense. This approach induces an (usual) gauge theory with a single true curvature. We have seen in this paper, that our approach with a higher gauger theory induces two fields (the fake curvature and the curving) which are complementary to describe non-adiabatic transition regions and kinematic decoherence. Our approach is able to interpret the dynamical phenomenon occurring during the control, and to discriminate purely non-adiabatic effects (similar to non-adiabatic effects of closed quantum systems, which are essentially associated with the fake curvature) from kinematic decoherence effects (hampering of the adiabatic control by the entanglement, which is essentially associated with the curving). A single gauge field as the one defined by the Hilbert-Schmidt analysis characterizes a mixing of these different dynamical processes, and does not provide clear interpretations.

\subsection{Conclusion}

The higher gauge structure (\cite{Baez}) associated with the $C^*$-geometric phases \cite{Viennot1} involves two fields. The average fake curvature is a measure of the non-adiabaticity of the system entangled with the environment (as the usual curvature for an isolated system). Nevertheless for invertible eigenmatrices (corresponding to a strong information lost associated with the partial trace $\tr_\E$) inactive singularities associated with the non-adiabaticity of the environment occurs. This is a drawback in the interpretation of this field. The average curving is a measure of the ``kinematic decoherence'' associated with variations of the entropy and of the entanglement between the system and the environment. It is essentially induced by the non-adiabaticity of the environment. In the example treated in this paper (the STIRAP system), it is zero for some invertible eigenmatrices. This seems to be the signature that the non-adiabaticities of the environment do not induce increases of the entanglement and of the entropy in these cases. But in these cases, the curving does not play its second role which consists to ``kill'' the inactive singularities in the fake curvature. Finally the von Neumann entropy of the eigenmatrice measures the ``local decoherence'' induced by the ``non-dynamical'' entanglement of the system with the environment appearing directly at the level of the eigenstates of the universe.\\
The method consisting to draw the densities of the different field strengths is able to interpret the hampering by entanglement on the adiabatic quantum control. It is also able to distinguish in the geometric representation of the universe dynamics, the non-adiabatic effects concerning the system from the decoherence effects associated with the environment. The study of the adiabatic fields can then be a tool to establish adiabatic control strategies for a system in contact with an environment. But as we can see for the example treated in the present paper (the STIRAP system entangled with another STIRAP system), it can be very difficult to completely avoid the decoherence regions (especially if an ``entropic string'' split into two parts the control manifold as it is the case for the example with a static coupling).

\end{document}